\documentclass[journal,10pt]{IEEEtran} 
\usepackage{dsfont}
\usepackage{amsthm}
\usepackage{amsmath}
\usepackage{amsfonts}
\usepackage{textcomp}
\usepackage{algorithm}
\usepackage{algpseudocode}
\usepackage{array}
\usepackage[caption=false,font=normalsize,labelfont=sf,textfont=sf]{subfig}
\usepackage{textcomp}
\usepackage{stfloats}
\usepackage{url}
\usepackage{verbatim}
\usepackage{graphicx}
\usepackage{cite}
\usepackage{setspace}
\usepackage{amsfonts}
\usepackage{mathtools}
\usepackage{braket}
\usepackage{bbm}
\newtheorem{theorem}{Theorem}

\algnewcommand{\algorithmicand}{\textbf{ and }}
\algnewcommand{\algorithmicor}{\textbf{ or }}
\algnewcommand{\OR}{\algorithmicor}
\algnewcommand{\AND}{\algorithmicand}

\usepackage{todonotes}
\usepackage{xcolor}

\usepackage{physics}
\usepackage[T1]{fontenc}

\begin{document}

\title{Entanglement Distribution Delay Optimization in Quantum Networks with Distillation}

\author{Mahdi~Chehimi,
        Kenneth~Goodenough,
        Walid~Saad,
        Don~Towsley,
        Tony~X.~Zhou,
\thanks{M. Chehimi and W. Saad are with the Department of Electrical and Computer Engineering, Virginia Tech University, Arlington, VA USA, e-mail: (mahdic@vt.edu, walids@vt.edu)}
\thanks{Kenneth Goodenough and Don Towsley are with the Manning College of Information and Computer Sciences, University of Massachusetts Amherst, USA. e-mail: (kdgoodenough@gmail.com, towsley@cs.umass.edu)}
\thanks{Tony X. Zhou is with the Research Laboratory of Electronics, Massachusetts Institute of Technology, Cambridge, MA 02139, USA. email: (xu.zhou@ngc.com.)}\vspace{-0.9cm}}


\maketitle
\begin{abstract}
Quantum networks (QNs) that distribute entangled states over optical channels enable secure distributed quantum computing and sensing applications over next-generation optical communication networks. However, quantum switches (QSs) in such QNs, which perform entanglement distribution, have limited resources, e.g., single-photon sources (SPSs) and quantum memories, which are highly sensitive to noise and losses. These QS resources must be carefully allocated to minimize the overall entanglement distribution delay. In this paper, a \emph{QS resource allocation framework} is proposed, which jointly optimizes the average \emph{entanglement distribution delay} and \emph{entanglement distillation operations}, to enhance the end-to-end (e2e) fidelity and satisfy users application-specific minimum rate and fidelity requirements. The proposed framework considers realistic QN noise and imperfections, and includes the derivation of the analytical expressions for the average quantum memory decoherence noise parameter, and the resulting e2e fidelity after distillation. Finally, the proposed framework accounts for practical QN deployment considerations, where QSs can control 1) \emph{nitrogen-vacancy (NV) center} SPS types based on their isotopic decomposition, and 2) \emph{nuclear spin} regions based on their distance and coupling strength with the electron spin of NV centers. A simulated annealing metaheuristic algorithm is proposed to solve the QS resource allocation optimization problem. Simulation results show that the proposed framework manages to satisfy all users rate and fidelity requirements, unlike existing \emph{distillation-agnostic (DA)}, \emph{minimal distillation (MD)}, and \emph{physics-agnostic (PA)} frameworks which do not perform distillation, perform minimal distillation, and does not control the physics-based NV center characteristics, respectively. Furthermore, the proposed framework results in around $30\%$ and $50\%$ reductions in the average e2e entanglement distribution delay compared to existing PA and MD frameworks, respectively. Moreover, the proposed framework results in around $5\%$, $7\%$, and $11\%$ reductions in the average e2e fidelity compared to existing DA, PA, and MD frameworks, respectively.\vspace{-0.1cm}
\end{abstract}

\begin{IEEEkeywords}
Entanglement distribution delay, quantum communications, quantum networks, Nitrogen-vacancy centers, quantum memory decoherence.
\end{IEEEkeywords}

\IEEEpeerreviewmaketitle
\vspace{-0.55cm}
\section{Introduction}\vspace{-0.1cm}
\IEEEPARstart{Q}{uantum} networks (QNs) represent a major leap in next-generation optical communication networks as they enable the distribution of entangled quantum states between distant nodes over optical channels. As such, QNs enable a large spectrum of quantum applications, such as distributed quantum computing \cite{cacciapuoti2019quantum}, quantum clock synchronization \cite{jozsa2000quantum}, and quantum key distribution (QKD) \cite{cao2022evolution}. In this regard, \emph{star-shaped QNs}, leveraging a central quantum switch (QS) serving a set of users, play a major role in simultaneously enabling different quantum applications that have heterogeneous requirements on both the \emph{entanglement rates} and the quality, or \emph{fidelity}, of distributed entangled qubits \cite{chehimi2022physics}.

In general, the QS is responsible for satisfying the received user requests for entangled qubits while meeting their quantum application-specific heterogeneous requirements. However, a QS has limited resources, such as single-photon sources (SPSs) to generate entangled qubits, and quantum memories to store them. These resources must therefore be carefully controlled and allocated to the different users. This QS resource allocation task faces a multitude of challenges stemming from the inherent quantum noise, losses, and \emph{decoherence} effects experienced by qubits in the QN, along with the probabilistic quantum operations that must be performed. 

As such, the QS needs to optimize its resources and allocate its SPSs and quantum memories to the users by capturing tradeoffs between their distances, quantum channel conditions, and their minimum entanglement rate and fidelity requirements \cite{lee2022quantum}. Simultaneously, the same QS resources are also used to perform entanglement distillation, necessary to enhance the quantum states fidelity. Accordingly, it is challenging for the QS to \emph{perform resource allocation} for entanglement distribution while \emph{optimizing the application of entanglement distillation operations}. Here, the QS must decide which distillation protocol to apply, the number of entangled qubits to distill, and which qubits to choose. When making such decisions on entanglement distillation, the QS faces another challenge in \emph{minimizing its overall entanglement distribution delay}, a generally overlooked problem. 

\vspace{-0.5cm}
\subsection{Related Works}\vspace{-0.1cm}
Some of the aforementioned challenges were separately addressed in prior works \cite{vardoyan2019stochastic,vasantam2021stability,vardoyan2023quantum,chehimi2021entanglement,pouryousef2023quantum,gauthier2023control,panigrahy2023capacity,chehimi2023matching,9351761,dai2020quantum}, but they have not been jointly addressed. For instance, the authors in \cite{vardoyan2019stochastic} performed a theoretical analysis of the stochastic performance and aggregate capacity of a QS. Furthermore, the authors in \cite{vasantam2021stability} investigated the entanglement distribution stability of a single QS serving different users with minimum fidelity requirements. However, the works in \cite{vardoyan2019stochastic} and \cite{vasantam2021stability} focused on the QS capacity and entanglement distribution stability, and did not consider optimizing the allocation of SPSs and quantum memories to different users with heterogeneous requirements. Additionally, the works in \cite{vardoyan2019stochastic} and \cite{vasantam2021stability} did not analyze the delay associated with the entanglement distribution process.

In \cite{vardoyan2023quantum}, the authors studied the QS resource allocation within a QN utility maximization framework whose goal is to optimize the entanglement generation rate-fidelity tradeoff. Moreover, the authors in \cite{chehimi2021entanglement} studied the problem of entanglement distribution rate maximization by optimizing the allocation of entanglement generation sources while considering limited memory capacity in a star-shaped QN. Additionally, the authors in \cite{pouryousef2023quantum} proposed an optimization framework for QN planning to maximize a quantum utility function based on the achieved end-to-end (e2e) rate and fidelity in the presence of quantum memory decoherence and multiplexing. Furthermore, the work in \cite{gauthier2023control} proposed a cost-effective control architecture that allows a QS to fairly generate and distribute entangled resources over users in a QN utility maximization framework. Although  \cite{vardoyan2023quantum,chehimi2021entanglement,pouryousef2023quantum,gauthier2023control} focused on the task of resource allocation in a QS, they did not take into account the possibility of performing entanglement distillation operations, a process that is necessary to enhance the distributed states fidelity. Another limitation of the works in \cite{vardoyan2023quantum,chehimi2021entanglement,pouryousef2023quantum,gauthier2023control} is that they did not optimize the resulting delay in the QS operation.

Meanwhile, the authors in \cite{panigrahy2023capacity} performed resource allocation and QS operations scheduling, particularly entanglement swapping and distillation, in a star-shaped QN with minimum fidelity constraints. The work in \cite{chehimi2023matching} used matching theory to perform request-QS assignment and entanglement swapping-distillation scheduling in a multi-QS star-shaped QN. While the works in \cite{panigrahy2023capacity} and \cite{chehimi2023matching} considered entanglement distillation operations by QSs, they did not analyze nor optimize the entanglement distribution delay associated with the QS operations. Meanwhile, the authors in \cite{9351761} and \cite{dai2020quantum} analyzed the entanglement distribution delay, but without considering entanglement distillation. Specifically, the work in \cite{dai2020quantum} studied the queuing delay in QN where a central node received entanglement requests according to a Poisson process and proposed a quantum memory control policy to minimize the resulting average queuing delay. However, the works in \cite{9351761} and \cite{dai2020quantum} did not account for noisy quantum gates and measurement operations, and did not consider the possibility for a QS to perform entanglement distillation.

While prior works \cite{vardoyan2019stochastic,vasantam2021stability,vardoyan2023quantum,chehimi2021entanglement,pouryousef2023quantum,gauthier2023control,panigrahy2023capacity,chehimi2023matching,9351761,dai2020quantum} have provided insightful results on separate aspects of QS resource allocation, like stability, memory management, and rate maximization, they have, generally, overlooked critical aspects, like delay and distillation optimization. In particular, \emph{no prior work has studied QS resource allocation while considering realistic noise sources and jointly optimizing the entanglement distribution delay and entanglement distillation operations.} 

\vspace{-0.3cm}
\subsection{Contributions}\vspace{-0.1cm}
The main contribution of this work is a QS resource allocation framework that jointly optimizes the average entanglement distribution delay and entanglement distillation operations while managing quantum memories to satisfy heterogeneous application-specific user requirements on minimum rate and fidelity. Towards achieving this goal, we make the following key contributions: 
\begin{itemize}
    \item We introduce a novel resource allocation framework for a single QS, equipped with limited resources, serving users with heterogeneous quantum application-specific requirements. For QS resources, although our framework can be generalized to any technology, we focus on the practical utilization of nitrogen-vacancy centers (NV centers) in synthetic diamond chips as SPSs and their nuclear spins as quantum memory, since it is one of the most successful technologies to practically deploy QNs \cite{doherty2013nitrogen}. Additionally, the QS is considered to be capable of performing entanglement distillation operations to enhance the e2e fidelity of entangled states. The framework keeps track of delays resulting from different entanglement generation attempts, transmission, operations, and quantum gates necessary to generate e2e entanglement, along with delays resulting from entanglement distillation. Thus, the QS captures the e2e entanglement distribution delay, which is optimized in the resource allocation process.  
    
    \item Our proposed framework accounts for loss, and quantum noise and imperfections during entanglement generation, transmission, and storage. In particular, we derive analytical expressions for the average quantum memory decoherence noise parameter experienced by the different qubits stored in nuclear spins during entanglement generation attempts for different distillation protocols and memory types. Then, we derive an expression for the fidelity of a distributed e2e state, along with its average experienced service delay.

    \item We consider practical QN deployment aspects through an example based on NV centers in diamond, where we equip the QS with control over the selection between two NV center types based on their isotopic decomposition, along with the selection between two nuclear spin regions based on their distance and coupling strength with the NV center's electron spin. By introducing these two binary control variables, we allow the QS to balance the different tradeoffs between the entanglement generation probability of success, coherence time, and quantum gate noise and speed.

    \item We formulate an optimization problem to minimize the average entanglement distribution delay while jointly optimizing the selection of entanglement distillation protocols, SPS type, and quantum memory characteristics allocated to different users. These users are assumed to have different application-specific minimum requirements on the average entanglement generation rate and fidelity. Then, we solve this problem using a simulated annealing metaheuristic algorithm.

\end{itemize}

Simulation results show that the proposed framework manages to satisfy all users heterogeneous minimum average e2e fidelity and rate requirements, unlike existing approaches which fail to do so when the distance between users and the QS increases. Furthermore, the proposed framework results in around $30\%$ reduction in the average e2e entanglement distribution delay and around $7\%$ reduction in the average e2e fidelity compared to existing \emph{physics-agnostic} frameworks that do not control the physics-based NV center characteristics accounted for in our framework. Additionally, the proposed framework results in more than $50\%$ reduction in the average e2e delay and up to $11\%$ reduction in the average e2e fidelity compared to existing frameworks that perform \emph{minimal entanglement distillation}, in addition to around $5\%$ reduction in the average fidelity of \emph{distillation-agnostic} existing frameworks. We also analyze the different relations between the NV center type, nuclear spin region, distillation protocol, and entanglement generation parameter.

The rest of this paper is organized as follows. Section \ref{sec_system_model} describes the star-shaped QN model and QS operations. Then, Section \ref{sec_noise} includes the derivation of average memory decoherence noise and e2e fidelity expressions. Section \ref{sec_optimization} formulates the QS resource allocation optimization problem minimizing average entanglement distribution delay and its solution, while Section \ref{sec_simulations} presents simulation results and discussions. Finally, conclusions are drawn in Section \ref{sec_conclusion}.
\vspace{-0.3cm}
\section{System Model}\label{sec_system_model}
Consider a star-shaped QN consisting of a QS that serves a set $\mathcal{V}$ of $V$ users through optical fibers. The users have heterogeneous requirements in terms of their minimum required average entanglement generation rate and their corresponding minimum average fidelity, which vary depending on their quantum applications. The main task of the QS is to distribute entangled pairs of qubits among different users subject to minimum rate and fidelity requirements while minimizing the average entanglement distribution delay. 

Towards this goal, the QS is equipped with multiple independent SPSs, represented by the set $\mathcal{M}$ of $M\leq V$ synthetic diamond chips, each of which is engineered to include a single, readily-prepared NV center\footnote{We only assume one NV center to be active inside each synthetic diamond to minimize interference and simplify the analysis.}. Similarly, each user is equipped with one synthetic diamond SPS having a single NV center. The choice of NV centers in synthetic diamonds is based on recent advances in their fabrication techniques and their practicality as they provide unparalleled control over their characteristics through doping, isotopic engineering, and defect management \cite{atature2018material}. Furthermore, an optical heralding station consisting of a balanced beam splitter and two single photon detectors is placed on the optical fiber between the QS and every user. The QS operates in a time-slotted manner, where a time slot is the duration of an entanglement generation attempt by an NV center, denoted as $T_{\mathrm{att}}$. We focus on optimizing the QS operation over an extended episode of duration $T_{\mathrm{e}}$, consisting of a large number of time slots $0,1,\ldots, T_{\mathrm{e}}$, which grows infinitely to ensure the applicability of the law of large numbers, thereby facilitating the analysis of average performance. Additionally, each NV center is dedicated to only one user during an operation episode. Hence, the QS can serve $M$ users simultaneously. Thus, we denote by $\mathcal{U}$ the set of $M$ users served during an episode, where each user is at a distance $d_i, \forall i\in\mathcal{U}$ from the QS. At the beginning of an episode, for each served user $i\in\mathcal{U}$, both its NV center and one NV center from the QS attempt to emit a single photon each, entangled to their corresponding electron spins. These photons are sent over optical fiber towards the heralding station, which performs entanglement swapping, or Bell state measurement (BSM) operations, to entangle the two remote electron spins, which corresponds to a QS-user, \emph{spin-spin}, entangled link.

\vspace{-0.3cm}
\subsection{NV Centers and Entanglement Generation}\label{sec_NV_centers_entanglement_generation}\vspace{-0.1cm}
NV centers represent one of the most effective physical platforms for practically implementing SPSs and quantum memories used in QNs due to their distinct characteristics, such as operating at room temperature and their extended coherence times. The formation of an NV center involves substituting one carbon atom in the diamond lattice with a nitrogen atom, a process facilitated by irradiation. This step is followed by annealing, which induces vacancies within the lattice structure \cite{doherty2013nitrogen}. Each NV center is considered to incorporate an electron spin qubit serving as an optical interface for single photon emission to generate spin-photon entanglement, and a group of additional carbon $^{13}C$ nuclear spin qubits functioning as quantum memories, along with abundant $^{12}C$ carbon isotopes that have a nuclear spin of zero and do not contribute to the magnetic noise that causes decoherence \cite{bernien2013heralded}.

Although synthetic diamonds provide a significant control over the NV center characteristics, it is practically challenging to achieve identical properties among a large number of NV centers within synthetic diamonds \cite{brandenburg2018improving}. To capture this practical challenge, we assume the presence of two NV center types, where each type has a different probability of successfully emitting an entangled photon. In particular, type-1 NV centers are designed to optimize photon emission efficiency, while type-2 NV centers represent NV centers where suboptimal conditions result in reduced emission probabilities. For instance, NV centers in type 1 diamonds can be engineered to have more $^{12}C$ carbon isotopes in the diamond lattice, which results in a calmer environment with longer coherence times, while type 2 diamonds have fewer $^{12}C$ isotopes \cite{balasubramanian2009ultralong}.  

Here, a user's NV center can be of either type, while the QS's $M$ NV centers include $M_1$ NV centers of type-1, and $M_2$ of type-2, where $M_1 + M_2 = M$. We introduce the QS binary control variable $x_i, \forall i\in\mathcal{U}$, which takes a value of $0$ when the NV center allocated to user $i\in\mathcal{U}$ is of type-1, and takes a value of $1$ otherwise. Inside each diamond, there exists a set of $^{13}C$ nuclear spins coupled to the NV center's electron spin via magnetic dipole-dipole interactions. Each of these nuclear spins can be a quantum memory that stores qubits for extended durations through dynamic decoupling. The studied system model with all its elements is shown in Fig. \ref{fig_architecture}.

\begin{figure}
\includegraphics[width=\columnwidth]{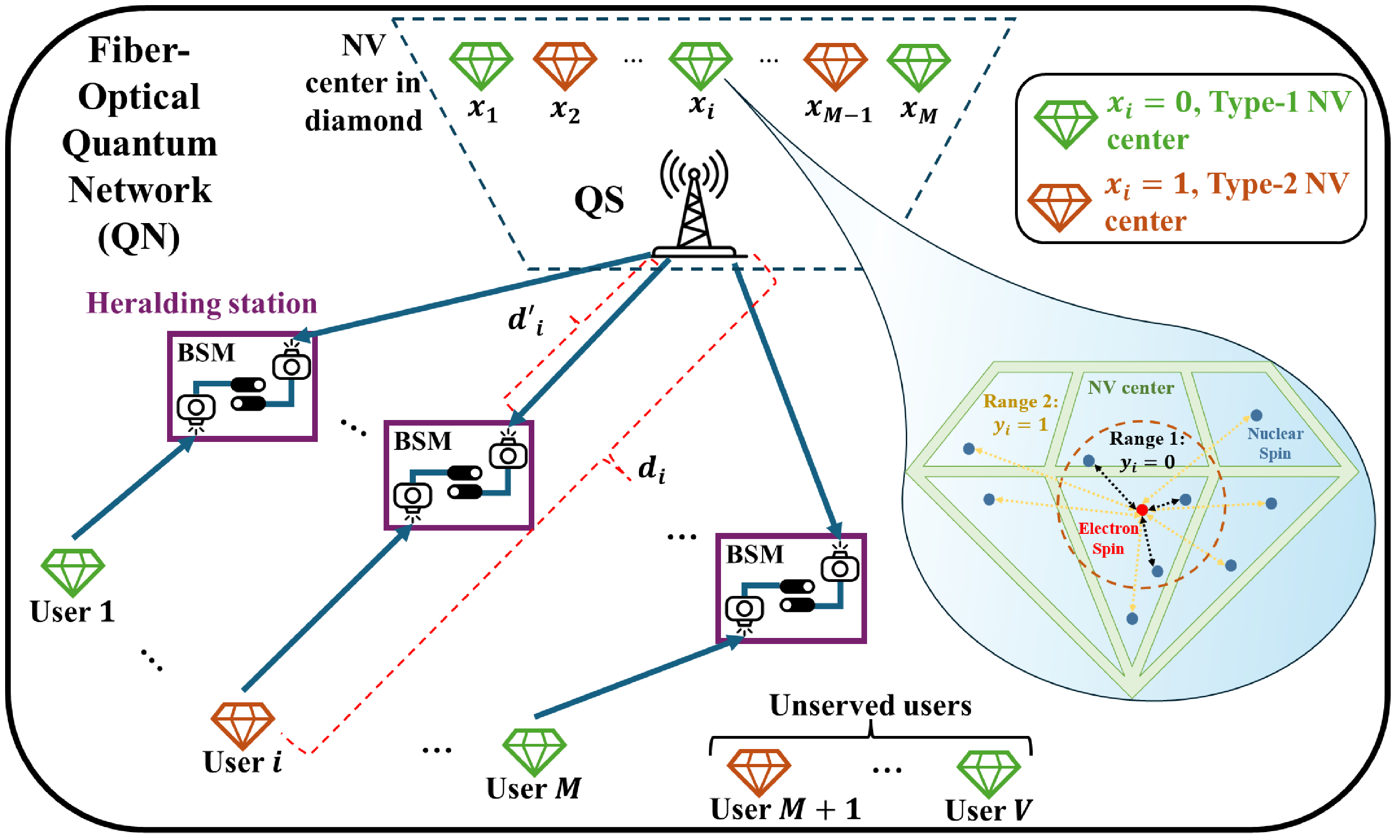}\vspace{-0.3cm}
\caption{The proposed QN system model.}
\label{fig_architecture}
\vspace{-0.6cm}
\end{figure}

On each NV center, selective microwave and optical pulses are applied to the electron spin in an attempt to emit a single photon, creating spin-photon entanglement. Note that for an emitted entangled photon to be \emph{useful}, it must fall within the zero-phonon line (ZPL) in its optical emission spectrum \cite{rozpkedek2019near}. Once an NV center successfully emits a single photon within the ZPL, the resulting ideal spin-photon entangled state, which entangles the spin with the presence or absence of an entangled photon, for user $i\in\mathcal{U}$ can be represented by:
\begin{equation}\label{eq_ideal_entangled_state}
    \ket{\psi_i^+} = \sin\theta_i\ket{\downarrow}\ket{0} + \cos\theta_i\ket{\uparrow}\ket{1}, 
\end{equation}
where $\ket{\uparrow}$ ($\ket{\downarrow}$) is the bright (dark) state of the electron spin qubit, and $\ket{1}$ ($\ket{0}$) refer to the presence (absence) of an emitted photon. Here, coherent microwave pulses can be applied to tune the parameter $\theta_i, \forall i\in\mathcal{U}$, enabling a tradeoff between entanglement generation rate and fidelity \cite{rozpkedek2019near}.


To create e2e entanglement over a link between the QS and a user, two NV centers, each at one end, attempt to generate spin-photon entanglement, and the two emitted photons are sent towards the heralding station. At this station, the two photons are overlapped to remove their which-path information. If a single photon is detected at the heralding station, then an e2e spin-spin entangled link is established between the NV centers at the QS and the respective user \cite{bernien2013heralded}.

\vspace{-0.45cm}
\subsection{Quantum Channel Losses and Noise}\vspace{-0.1cm}
The process of establishing an e2e entangled link is exposed to several sources of noise and losses due to imperfections in the entangled photon emission, collection, transmission, detection, and measurement operations. For instance, for each NV center, the probability of successfully emitting a single entangled photon that falls within the ZPL, $P_\mathrm{e}$, is affected by the type of NV center, i.e., a type-1 NV center has a higher value of $P_\mathrm{e}$. Further, if such a photon is successfully emitted, it is collected from the diamond with probability $P_{\mathrm{ce}}$, which captures the collection efficiency \cite{bernien2013heralded}. 


Thus, the operation of preparing an entangled state, as described by \eqref{eq_ideal_entangled_state}, is not ideal and it can be modeled by a quantum dephasing noise channel, with parameter $\lambda_{\mathrm{prep}}$. Consequently, a prepared spin-photon entangled state, for a user $i\in\mathcal{U}$, is not purely $\ket{\psi_i^+}\bra{\psi_i^+}$, but rather a mixture of the form $\boldsymbol{\rho}_{\mathrm{prep}}(\theta_i) = \lambda_{\mathrm{prep}}\ketbra{\psi_i^+}{\psi_i^+} + (1-\lambda_{\mathrm{prep}})(\mathbb{I}\otimes Z)\ket{\psi_i^+}\bra{\psi_i^+}(\mathbb{I}\otimes Z) = \lambda_{\mathrm{prep}}\ket{\psi_i^+}\bra{\psi_i^+} + (1-\lambda_{\mathrm{prep}}) \ket{\psi_i^-}\bra{\psi_i^-},$
where $\ket{\psi_i^-} = \sin\theta_i\ket{\downarrow}\ket{0} - \cos\theta_i\ket{\uparrow}\ket{1}$.

After preparing such a spin-photon entanglement on both the QS and user sides, the two emitted photons are sent over optical fiber towards the intermediate heralding station, which performs BSM. Throughout this transmission process, and before the BSM, the photons experience several losses due to their interactions with the surrounding environment, which scale exponentially with the travelled distance. Finally, the single-photon detectors on the heralding station have a probability, $P_\mathrm{det}$, of successfully detecting a transmitted photon that have arrived at the detector. 

To simplify the analysis, we assume that photons sent towards the heralding station from both the QS and users experience the same loss probability. Since all the aforementioned losses are independent, the overall experienced loss parameter between the QS and user $i\in\mathcal{U}$, before performing a BSM at the heralding station, is:
\begin{equation}
    \eta(x_i) = P_\mathrm{e,QS}(x_i) \times P_{\mathrm{ce}}\times e^{-\frac{d'_i(x_i)}{L_0}} \times P_\mathrm{det}, 
\end{equation}
where $L_0$ is the attenuation length of the optical fiber and
\begin{equation}\footnotesize
    d'_i(x_i) =
    \begin{cases}
    d_i/2  & \quad P_\mathrm{e,QS}(x_i) = P_\mathrm{e,user}, \\
    \frac{1}{2}\left[d_i - L_0\ln\left(\frac{P_\mathrm{e,QS}(x_i)}{P_\mathrm{e,user}}\right) \right]    & \quad P_\mathrm{e,QS}(x_i) \neq P_\mathrm{e,user}
    \end{cases},
\end{equation} where $d'_i(x_i)$ is the distance between the QS and the heralding station, chosen so as to enforce the aforementioned symmetry of losses. This pure loss channel can be modeled as an \emph{amplitude-damping noise channel} affecting the entangled photon, with damping parameter $1-\eta(x_i)$ \cite{chuang1997bosonic}. Furthermore, the lack of knowledge about the relative phase between the two photons sent to the heralding station is modeled as a \emph{dephasing noise channel}, with a dephasing parameter $\lambda_{\mathrm{ph}} = \frac{I_1 \left(\frac{1}{(\Delta\phi)^2}\right)}{2I_0 \left(\frac{1}{(\Delta\phi)^2}\right)} + \frac{1}{2}$. Here, $I_0 (I_1)$ is the Bessel function of order zero (one), and $\Delta\phi$ is the uncertainty in the phase acquired by the photons during their transmission. 

After detecting the photons and removing their which-path information at the heralding station, a BSM is performed, and the resulting initial e2e spin-spin entangled link between the QS and user $i\in\mathcal{U}$ is
\begin{equation}\label{eq_rho_initial_spin_spin}
\begin{split}
\boldsymbol{\rho}_{\mathrm{in}}&(\theta_i,x_i) =  \frac{2 \sin^2(\theta_i)}{2 - \eta(x_i) \cos^2(\theta_i)} (a \ketbra{\Psi^+}{\Psi^+} \\
&+ b \ketbra{\Psi^-}{\Psi^-}) + \frac{\cos^2(\theta_i)(2 - \eta(x_i))}{2 - \eta(x_i) \cos^2(\theta_i)} \ketbra{\uparrow\uparrow}{\uparrow\uparrow},
\end{split}
\end{equation}
where $\ket{\Psi^{\pm}} = \frac{1}{\sqrt{2}}(\ket{\downarrow\uparrow} \pm \ket{\uparrow\downarrow})$, $a = \lambda_{\mathrm{ph}} \left( \lambda^2_{\mathrm{prep}} + (1 - \lambda_{\mathrm{prep}})^2 \right) + 2 \lambda_{\mathrm{prep}} (1 - \lambda_{\mathrm{prep}})(1 - \lambda_{\mathrm{ph}}),$ and $b = (1 - \lambda_{\mathrm{ph}}) \left( \lambda^2_{\mathrm{prep}} + (1 - \lambda_{\mathrm{prep}})^2 \right) + 2 \lambda_{\mathrm{prep}} (1 - \lambda_{\mathrm{prep}}) \lambda_{\mathrm{ph}}$. Moreover, the desired Bell state is considered to be $\ket{\Psi^+}$, and the probability of successfully establishing such an e2e spin-spin entangled link is given by \cite{rozpkedek2019near}:
\begin{equation}\label{eq_prob_succ_single_spin_spin_entanglement}
    P_{\mathrm{in}}(\theta_i,x_i) = 2\eta(x_i)\cos^2(\theta_i)\left(1-\frac{\eta(x_i)}{2}\cos^2(\theta_i)\right).
\end{equation}

\vspace{-0.5cm}
\subsection{QS Quantum Memory Selection}\vspace{-0.1cm}
At each NV center, after the successful generation of a spin-photon entanglement, the electron spin qubit is moved for storage to one of the available nuclear spins inside the diamond. This frees the electron spin to participate in new entanglement generation attempts, and ensures that the qubit has a long coherence time, since the lifetime of nuclear spins can be significantly enhanced through dynamic decoupling. In general, a large number of $^{13}C$ nuclear spins exist inside each diamond. The density of nuclear spins may vary depending on the engineering complexity of the diamond growth. These nuclear spins are non-uniformly distributed around the electron spin inside a 3-dimensional diamond lattice. However, through dynamic decoupling and precise microwave pulse sequences, we can enable the selective manipulation and coherence protection of nuclear spins inside the diamond lattice \cite{fuchs2011quantum}. 

The effectiveness of a nuclear spin as a quantum memory is affected by its distance from the electron spin. This is due to the dipole-dipole interaction decreasing with the cube of the distance, $r$, between the spins ($1/r^3$). Typically, the relevant distances in diamond NV centers that affect coupling strength are on the order of a few angstroms (\AA\hspace{0cm}) to several nanometers given the size of the diamond lattice and the nuclear spins bath around the NV center \cite{abragam1961principles}. Hence, the distance between a nuclear spin and the electron spin of an NV center directly affects their coupling strength, and correspondingly impacts both the nuclear spin coherence time and the efficiency of quantum gates applied on it.

In particular, nuclear spins closer to the electron spin experience strong coupling, which enhances the fidelity of quantum gate applied on them due to reduced execution times. However, this strong coupling can reduce the coherence time of these nuclear spins. This is because the electron spin, which has a large magnetic moment compared to nuclear spins, decoheres quickly due to its strong interaction with the environment and magnetic fields, and this can significantly affect nuclear spins that are highly coupled to it \cite{bar2013solid}. In contrast, a nuclear spin that is far away from the electron spin, to an extent that still allows for control and readout, has weak coupling and generally exhibits a longer coherence time. In particular, such a nuclear spin is less affected by the electron's interactions with the external environment and by magnetic field fluctuations and interactions with other spins. However, the reduced coupling leads to lower fidelity and longer execution times for quantum gates applied on stored qubits (e.g., needed in tasks like nuclear spin initialization, control, and readout) \cite{taminiau2012detection}.


When choosing a nuclear spin to store entangled qubits with user $i\in\mathcal{U}$, the QS must consider the distance between the nuclear spins and the NV center's electron spin, as well as corresponding tradeoffs between coherence time, and quantum gate speed and fidelity. For tractability, we capture these tradeoffs by considering a simplified model of the diamond lattice, where we divide the nuclear spin bath surrounding the electron spin into two distinct regions based on the distance between nuclear spins and the electron spin. Moreover, we assume all nuclear spins within each region exhibit similar characteristics. In particular, nuclear spins in \emph{region 1} are closer to the electron spin, thus exhibiting a strong coupling and shorter coherence times. Furthermore quantum gates applied on them have less noise and shorter execution times. On the other hand, nuclear spins in \emph{region 2} are farther away from the electron spin, thus having a weak coupling and longer coherence times; quantum gates applied to them are noisier and incur longer execution times. Thus, we introduce a binary control variable $y_i, \forall i\in \mathcal{U}$, that corresponds to the selection of the region in which the nuclear spins selected to store entangled qubits for each user exist around the electron spin. Here, a value of $0$ corresponds to region 1, while a value of $1$ corresponds to region 2.
 
Accordingly, the coherence time of a nuclear spin depends on the region in which it exists, in addition to the NV center type, defined in Section \ref{sec_NV_centers_entanglement_generation}. Thus, the coherence time, $T_{\mathrm{c}}$, is a function of variables $x_i$ and $y_i, \forall i\in \mathcal{U}$. Additionally, for the parameters we will use later in our simulations, we consider $T_{\mathrm{c}}(x_i = 0, y_i = 1) \geq T_{\mathrm{c}}(x_i = 1, y_i = 1) \geq T_{\mathrm{c}}(x_i = 0, y_i = 0) \geq T_{\mathrm{c}}(x_i = 1, y_i = 0)$. Here, the distance between the nuclear and electron spins is assumed to have a more severe impact on coherence time than the NV center type, or isotopic composition, \cite{balasubramanian2009ultralong,maurer2012room}. Furthermore, for each user $i\in\mathcal{U}$, the average time needed to perform a 2-qubit gate on a nuclear spin, represented by $T_{\mathrm{gate}}(y_i)$, and the fidelity of such gates, denoted by $\alpha_{\mathrm{gate}}(y_i)$ both depend on the nuclear spin region. As discussed earlier, $T_{\mathrm{gate}}(y_i=0)\leq T_{\mathrm{gate}}(y_i=1)$, and $\alpha_{\mathrm{gate}}(y_i=0)\geq \alpha_{\mathrm{gate}}(y_i=1)$. The noise encountered during the entangled qubit generation, storage, transmission, and quantum gate application reduces the fidelity of e2e entangled qubits. Accordingly, the QS can perform entanglement distillation to enhance the e2e fidelity, as discussed next. 

\subsection{QS Entanglement Distillation Protocols}\label{sec_distillation}
In our QN model, the QS and users can apply entanglement distillation protocols based on $[n,k,d]$ encoding schemes. Here, $n$ qubit pairs are used to encode $k$ pairs, where $d$ is the distance of the code, and the resulting fidelities of the $k$ qubit pairs increase with the number of encoding pairs $n$  \cite{jansen2022enumerating}. Such protocols are derived in one-to-one correspondence with stabilizer codes, and can be implemented with Clifford circuits \cite{gottesman1997stabilizer}. We adopt the efficient entanglement distillation protocols based on optimized quantum Clifford circuits, developed in \cite{jansen2022enumerating}, where $n$ pairs are distilled to a single pair ($k=1$). The adopted protocols do not require multiple rounds of classical communications to inform the QS about the success or failure of each distillation attempt, which reduces the average QN entanglement distribution delay.

To equip the QS with the choice of different entanglement distillation circuits for different users based on their e2e fidelity requirements, we introduce an integer control variable, $z_i, \forall i\in \mathcal{U}$, that corresponds to the QS's selection of which distillation circuit to adopt for each user $i\in \mathcal{U}$. In particular, we consider seven possible values for the variable $z_i$, which correspond to: 1) $z_i = 1$: no distillation, 2) $z_i=2$: [2,1,1] distillation encoding scheme, and so on, until 7) $z_i = 7$: [7,1,3] distillation encoding scheme. Beyond $n = 7$ encoding pairs, the analysis and derived expressions for fidelity and probability of success become largely complex and intractable \cite{jansen2022enumerating}.


The probability of success of the different distillation schemes and the resulting distilled fidelity are derived in \cite{jansen2022enumerating}. As an example, when $z_i = 4$, the probability of successfully distilling entanglement for user $i\in\mathcal{U}$ is  \cite{jansen2022enumerating}
\begin{equation}\label{eq_distillation_p_success}\footnotesize
    P_{\mathrm{dis}}(F_{\mathrm{input}},z_i=4) = \frac{32}{27}F_{\mathrm{input}}^4 - \frac{4}{9}F_{\mathrm{input}}^2 + \frac{4}{27}F_{\mathrm{input}} + \frac{1}{9}, 
\end{equation}
where $F_{\mathrm{input}}$ represents the fidelity of the $z_i$ stored entangled qubit pairs for user $i\in\mathcal{U}$ before applying the entanglement distillation protocol. Similarly, for the case of $z_i = 4$, the resulting fidelity after distillation for user $i\in\mathcal{U}$ is \cite{jansen2022enumerating}: 
\begin{equation}\label{eq_distillation_protocol_fidelity}\footnotesize
     F_{\mathrm{dis}}(F_{\mathrm{input}},z_i=4) = \frac{\frac{8}{9}F_{\mathrm{input}}^4 + \frac{8}{27}F_{\mathrm{input}}^3 - \frac{2}{9}F_{\mathrm{input}}^2 + \frac{1}{27}}{\frac{32}{27}F_{\mathrm{input}}^4 - \frac{4}{9}F_{\mathrm{input}}^2 + \frac{4}{27}F_{\mathrm{input}} + \frac{1}{9}}.
\end{equation}


For each entanglement distillation Clifford circuit, based on $z_i$, the average gate noise associated with the circuit is modeled as a function of the number of its 2-qubit gates. This is because 2-qubit gates are significantly noisier than single-qubit gates. For each user $i\in\mathcal{U}$, the number of 2-qubit gates, $N_{\mathrm{gate}}$, in each Clifford circuit of the considered distillation schemes ($z_i=1,\ldots,7$) are 0, 1, 2, 4, 7, 8, 11 gates respectively. Here, we assume a simplified model where the average 2-qubit gate noise scales exponentially with $N_{\mathrm{gate}}$, as follows:  
\begin{equation}\label{eq_gate_noise_parameter}
    \lambda_{\mathrm{gate}}(y_i,z_i) = \left(\alpha_{\mathrm{gate}}(y_i)\right)^{N_{\mathrm{gate}}(z_i)}.
\end{equation}

Furthermore, the average time needed to run the entanglement distillation gates and circuits for user $i\in\mathcal{U}$ is assumed to be proportional to the depth of those circuits, $N_{\mathrm{depth}}(z_i)$, and the average time needed to perform a 2-qubit gate, $T_{\mathrm{gate}}(y_i)$, on qubits stored in nuclear spins in region $y_i$. The distillation circuit depths, $N_{\mathrm{depth}}(z_i)$, for the different protocols ($z_i=1,\ldots,7$) are \cite{jansen2022enumerating} 0, 1, 2, 3, 5, 6, 6, respectively. Accordingly, the average time to perform entanglement distillation is
\begin{equation}
    T_{\mathrm{dis}}(y_i,z_i) = N_{\mathrm{depth}}(z_i)\times T_{\mathrm{gate}}(y_i).
\end{equation}

After successfully performing entanglement distillation, the resulting e2e entangled states are ready to be used by the users for their quantum applications. However, in order to use those resulting states in the intended quantum applications, they must satisfy a user-specific minimum fidelity requirement. Thus, it is necessary to analyze the impact of the heterogeneous quantum noise on the fidelity of the resulting e2e entangled states and to derive an expression for e2e fidelity, which we delve into next.

\vspace{-0.2cm}
\section{Average E2E Noise and Entanglement Fidelity}\label{sec_noise}\vspace{-0.05cm}
As mentioned earlier, we will be optimizing the average entanglement distribution delay over an extended episode of infinitely large duration $T_{\mathrm{e}}$. Each single e2e spin-spin entanglement generation attempt between the QS and a user $i\in \mathcal{U}$, located at distance $d_i$ from the QS, succeeds with probability $P_{\mathrm{in}}(\theta_i,x_i)$. When a single e2e entangled link is successfully generated between the QS and user $i\in\mathcal{U}$, its density matrix is given by $\boldsymbol{\rho}_{\mathrm{in}}(\theta_i,x_i)$, as in \eqref{eq_rho_initial_spin_spin}. Then, based on the selected entanglement distillation protocol, the QS must ensure the successful generation of $z_i$ such e2e entangled links between the QS and user $i\in\mathcal{U}$. In particular, the QS and user continue to attempt entanglement generation until they obtain $z_i$ successes. The resulting $z_i$ pairs must be stored in the quantum memory of the QS and user $i\in\mathcal{U}$, inside the $y_i$ selected nuclear spin region. Here, the number of attempts, $N_{\mathrm{att}}(z_i)$ until obtaining $z_i$ successes can be described by a \emph{negative binomial} distribution, i.e., $N_{\mathrm{att}}(z_i)\sim \text{NB}(z_i, P_{\mathrm{in}})$, whose probability mass function is given by $P(N_{\mathrm{att}}(z_i) = k) = \binom{k+z_i-1}{z_i-1} P_{\mathrm{in}}^{z_i} (1 - P_{\mathrm{in}})^{k}, \quad k = 0, 1, 2, \ldots$. 

The $z_i$ entangled pairs are generated at different times. Thus they suffer different amounts of decoherence. Hence, we introduce $\boldsymbol{S} = [t_1, t_2, ..., t_{z_i}]$, where $t_k$ denotes the time the $k$-th spin-spin e2e entangled link is established, $k=1,\ldots ,z_i$. Note that it is of the form of a \emph{strictly increasing sequence (SIS)} of $z_i$ items, i.e. $t_k < t_{k+1}$. Here, $t_{z_i}$ corresponds to the last successful attempt after which Clifford circuits are applied on the stored qubits at both the QS and corresponding user to perform entanglement distillation. 

To characterize the fidelity of the resulting e2e entangled states after distillation, we must capture the heterogeneous memory decoherence effects experienced by the different stored qubits. Hence, we calculate the average fidelity after obtaining $z_i$ successes in the entanglement generation attempts. To do so, we derive, next, the analytical expression for the average quantum memory decoherence noise experienced before performing distillation. Then, we derive the resulting e2e entangled state fidelity expression after distillation, which must satisfy the users' application-specific minimum requirements. 

\vspace{-0.4cm}
\subsection{Average Quantum Memory Decoherence Noise}\vspace{-0.1cm}
We model the quantum memory decoherence noise as a general, worst case, quantum depolarizing noise channel. Thus, for every user $i\in\mathcal{U}$, each of the $z_i$ stored qubits, after undergoing a depolarizing noise channel with a depolarizing parameter $j\in\{1,2,...,z_i\}$, will be represented by: $\boldsymbol{\rho}_{\lambda_{\mathrm{d},j}}(\theta_i,x_i,y_i) = \lambda_{\mathrm{d},j}(x_i,y_i)\boldsymbol{\rho}_{\mathrm{in}}(\theta_i,x_i) + \left(1-\lambda_{\mathrm{d},j}(x_i,y_i)\right)\frac{\boldsymbol{\Pi}}{4}$, where every stored qubit experiences a different depolarizing noise parameter based on its storage time. However, having such non-uniform quantum states makes the analysis intractable.

To simplify the analysis, we apply a twirling map that maps the different non-uniform quantum states into a unified quantum state with a uniform depolarizing noise parameter equal to the average of the $z_i$ non-uniform noise parameters $\lambda_{\mathrm{d},j}(x_i,y_i), \forall j\in\{1,2,...,z_i\}$ and $i\in\mathcal{U}$. A possible such twirling map is given by randomly applying a unitary from the group generated by a cyclic shift on the stored pairs. 



Applying a twirling map in our case can be interpreted as if a permutation is performed on the stored qubits, without knowing which specific permutation it is. While this corresponds to throwing away some information, note that twirling maps are mere mathematical tools which significantly ease the analytical analysis, without affecting the obtained fidelity expressions. In our case, if the QS performs a uniformly random permutation on stored qubits, then all stored states will become uniform, and each of the non-uniform stored entangled states $\boldsymbol{\rho}_{\lambda_{\mathrm{d},j}}(\theta_i,x_i,y_i)$ becomes: $\boldsymbol{\rho}_{\lambda_{\mathrm{d,avg}}}(\theta_i,x_i,y_i,z_i) = \lambda_{\mathrm{d,avg}}(x_i,y_i,z_i)\boldsymbol{\rho}_{\mathrm{in}}(\theta_i,x_i,y_i) + \left(1-\lambda_{\mathrm{d,avg}}(x_i,y_i,z_i)\right)\frac{\Pi}{4}$, where $\lambda_{\mathrm{d,avg}}(x_i,y_i,z_i) = \frac{\sum_{j\in\{1,2,...,z_i\}} \lambda_{\mathrm{d},j}(x_i,y_i,z_i)}{z_i}, \forall i\in\mathcal{U}$ is the average noise parameter of the previously non-uniform entangled states. Moreover, if user $i\in\mathcal{U}$, has a memory decoherence time (storage time) $T_{\mathrm{c}}(x_i,y_i)$, then for \emph{a single SIS}, $\boldsymbol{S} = [t_1, t_2, ..., t_{z_i}]$, i.e., the average memory decoherence depolarizing noise parameter for the $z_i$ stored qubits is 
\vspace{-0.2cm}\begin{equation}\label{eq_noise_single_SIS}\small
    \lambda_{\mathrm{d,avg}}(x_i,y_i,z_i) = \sum_{t_k \in \boldsymbol{S}}\frac{e^{\frac{-(t_{z_i} - t_k)}{T_c(x_i,y_i)}}}{z_i},
\end{equation}
and the resulting fidelity from this noise parameter will be: \vspace{-0.2cm}
\begin{equation}\label{eq_fidelity_single_SIS}\footnotesize
    F_{\mathrm{d,avg},i}(x_i,y_i, z_i) = \frac{3\times \lambda_{\mathrm{d,avg}}(x_i,y_i,z_i) + 1}{4}.
\end{equation}


However, we must account for \emph{all possible SISs} in the considered large episode of time to get a more representative noise expression that captures the average QS memory performance. Thus, we derive, in Theorem \ref{theorem_avg_noise}, the exact expression for the average memory decoherence depolarizing noise parameter averaged over all possible SISs in the QN, denoted by $\langle \lambda_{\mathrm{d,avg}}\rangle(x_i,y_i,z_i), \forall i\in\mathcal{U}$, as follows:
\begin{theorem}\label{theorem_avg_noise}
The QS average quantum memory decoherence depolarizing noise parameter affecting $z_i$ entangled qubits stored in nuclear spins in region $y_i$ in a type-$x_i$ NV center in diamond for serving user $i\in\mathcal{U}$ averaged over all possible SISs is:
\begin{equation}\label{eq_average_decoherence_noise_parameter}\footnotesize
        \langle \lambda_{\mathrm{d,avg}} \rangle(x_i,y_i,z_i) = \frac{1-R_i q_i}{\left(1-R_i\right)z_i}\left(1-\left(\frac{R_i - R_i q_i}{1-R_i q_i}\right)^{z_i}\right),
\end{equation}
where $R_i \equiv e^{-\frac{1}{T_{\mathrm{c}}(x_i,y_i)}}$ and $q_i \equiv 1 - P_{\mathrm{in}}(\theta_i,x_i)$. Here, $P_{\mathrm{in}}(\theta_i,x_i)$ is the probability of successfully establishing an e2e spin-spin entangled link between the QS and user $i\in\mathcal{U}$ at distance $d_i$ from the QS. Moreover, $T_{\mathrm{c}}(x_i,y_i)$ represents the quantum memory coherence time for user $i\in\mathcal{U}$ based on its selected NV center type and nuclear spin region.
\end{theorem} 

\begin{proof}
See Appendix \ref{appendix_average_noise_derivation}.
\end{proof}

The corresponding average fidelity for each of the $z_i$ stored entangled qubits, averaged over all possible SISs can be directly obtained from $\langle \lambda_{\mathrm{d,avg}} \rangle(x_i,y_i,z_i)$ as follows:
\begin{equation}\label{eq_avg_fidelity}\small
    \langle F_{\mathrm{d,avg}}\rangle (x_i,y_i,z_i) = \frac{3\times \langle \lambda_{\mathrm{d,avg}} \rangle(x_i,y_i,z_i) + 1}{4}.
\end{equation}

\vspace{-0.4cm}
\subsection{E2E Entangled State Fidelity}\vspace{-0.1cm}
Now that we have analyzed the different noise sources in the QN, we study the evolution of e2e spin-spin entangled quantum states from their initial generation, through their storage in quantum memories, and final participation in entanglement distillation subject to quantum gate noise and imperfections. We focus on the resulting fidelity of the final entanglement to ensure the satisfaction of user application-specific minimum fidelity requirements. 

In particular, when a single spin-spin entangled state is successfully generated, its density matrix is given by $\boldsymbol{\rho}_{\mathrm{in}}(\theta_i,x_i)$, described in \eqref{eq_rho_initial_spin_spin}. Here, the desired Bell state to be shared between the QS and user $i\in\mathcal{U}$ is considered to be $\ketbra{\Psi^{+}}{\Psi^{+}}$. Then, based on the selected entanglement distillation protocol, $z_i$ entangled qubit pairs are stored in quantum memories of the QS and corresponding user $i\in\mathcal{U}$. As mentioned earlier, using a twirling map, the heterogeneous decoherence effects experienced by the $z_i$ entangled pairs can be represented by a unified depolarizing noise channel with a single average noise parameter, that is averaged over all possible SISs, and is given by $\langle \lambda_{\mathrm{d,avg}} \rangle(x_i,y_i,z_i)$ in \eqref{eq_average_decoherence_noise_parameter}. The stored entangled pairs after capturing quantum memory decoherence effects can be found by applying the depolarizing quantum noise channel on \eqref{eq_rho_initial_spin_spin}, which results in $\langle\boldsymbol{\rho'}_{\mathrm{in}}\rangle(\theta_i,x_i,y_i,z_i) = \langle \lambda_{\mathrm{d,avg}} \rangle(x_i,y_i,z_i)\boldsymbol{\rho}_{\mathrm{in}}(\theta_i,x_i) + (1 - \langle \lambda_{\mathrm{d,avg}} \rangle(x_i,y_i,z_i))\frac{\boldsymbol{\Pi}}{4}$. Here, since the desired Bell state is $\ket{\Psi^{-}}$, the average fidelity of the resulting states is 
\begin{equation}\footnotesize\label{eq_F_in_prime}
    \begin{split}
    &\langle F'_{\mathrm{in}}\rangle(\theta_i,x_i,y_i,z_i) = \bra{\Psi^{+}}\langle\boldsymbol{\rho'}_{\mathrm{in}}\rangle(\theta_i,x_i,y_i,z_i)\ket{\Psi^{+}}\\ &= \frac{2a \langle \lambda_{\mathrm{d,avg}} \rangle(x_i,y_i,z_i) \sin^2(\theta_i)}{2 - \eta(x_i) \cos^2(\theta_i)} + \frac{(1-\langle \lambda_{\mathrm{d,avg}} \rangle(x_i,y_i,z_i))}{4}.
    \end{split}
\end{equation}

Then, in order to apply the considered entanglement distillation circuits, proposed in \cite{jansen2022enumerating}, on the $z_i$ stored entangled pairs, we observe that the analytical expressions for the resulting fidelity and probability of success for these distillation schemes require the stored entangled pairs to be in the form of Werner states, which can be achieved by applying another twirling map operation on the stored pairs. Here, note that the fidelity of a twirled state equals the fidelity of the state before twirling \cite{bennett1996mixed}, given by $\langle F_{\mathrm{W}}\rangle (\theta_i,x_i,y_i,z_i) = \langle F'_{\mathrm{in}}\rangle(\theta_i,x_i,y_i,z_i)$. As such, the resulting Werner state can be written as: $\langle\boldsymbol{\rho}_{\mathrm{W}}\rangle (\theta_i,x_i,y_i,z_i) = \langle F_{\mathrm{W}}\rangle (\theta_i,x_i,y_i,z_i)\ketbra{\Psi^{-}}{\Psi^{-}} + \left(\frac{1-\langle F_{\mathrm{W}}\rangle (\theta_i,x_i,y_i,z_i)}{3}\right)\left(\boldsymbol{\Pi}-\ketbra{\Psi^{-}}{\Psi^{-}}\right)$.

When a distillation protocol is applied, the stored entangled qubits, represented by Werner states $\langle\boldsymbol{\rho}_{\mathrm{W}}\rangle (\theta_i,x_i,y_i,z_i)$, are affected by the applied noisy quantum gates and operations for the distillation protocol. The quantum gate noise is modeled as a depolarizing quantum noise channel with noise parameter $\lambda_{\mathrm{gate}}(y_i,z_i)$, in \eqref{eq_gate_noise_parameter}, for user $i\in\mathcal{U}$. A stored averaged Werner state $\langle\boldsymbol{\rho}_{\mathrm{W}}\rangle (\theta_i,x_i,y_i,z_i)$, after undergoing this noise channel, becomes $\langle\boldsymbol{\rho'}_{\mathrm{W}}\rangle (\theta_i,x_i,y_i,z_i) = \left(\lambda_{\mathrm{gate}}(y_i,z_i)\right)\langle\boldsymbol{\rho}_{\mathrm{W}}\rangle (\theta_i,x_i,y_i,z_i) + (1-\lambda_{\mathrm{gate}}(y_i,z_i))\left(\frac{\boldsymbol{\Pi}}{4}\right)$, and its average fidelity becomes: 
\begin{equation}\label{eq_average_fidelity_before_distillation}\small
\begin{split}
    \langle F'_{\mathrm{W}}\rangle (\theta_i,x_i,y_i,z_i)& = \lambda_{\mathrm{gate}}(y_i,z_i)\langle F_{\mathrm{W}}\rangle (\theta_i,x_i,y_i,z_i)\\ 
    &+ \frac{(1-\lambda_{\mathrm{gate}}(y_i,z_i))}{4}.
\end{split}
\end{equation} 


Now, calculating the exact resulting average fidelity and probability of success after distillation, by averaging the results of $P_{\mathrm{dis}}(F_{\mathrm{input}},z_i)$ and $F_{\mathrm{dis}}(F_{\mathrm{input}},z_i)$ over all possible SISs is intractable. Hence, to simplify the analysis, we make an approximation by calculating the e2e fidelity after distillation, $F_{\mathrm{dis}}(F_{\mathrm{input}},z_i)$, for an input fidelity that is already averaged over all possible SIS, i.e., $F_{\mathrm{dis}}(\langle F_{\mathrm{input}}\rangle,z_i)$ instead of averaging $\langle F_{\mathrm{dis}}(F_{\mathrm{input}},z_i)\rangle$ itself over all possible SISs. We make a similar approximation for averaging the distillation probability of success expression over all possible SISs, and the two approximations made are summarized as: 

\begin{subequations}\label{eq_approximation}\small
    \begin{equation}
    \langle F_{\mathrm{dis}}(F_{\mathrm{input}},z_i)\rangle \approx F_{\mathrm{dis}}(\langle F_{\mathrm{input}}\rangle,z_i),
    \end{equation}
    \begin{equation}\small
    \langle P_{\mathrm{dis}}(F_{\mathrm{input}},z_i)\rangle \approx P_{\mathrm{dis}}(\langle F_{\mathrm{input}}\rangle,z_i),
    \end{equation}
\end{subequations}
where the accuracy of this approximation is verified through Monte Carlo simulations over one million possible SIS instances, see Appendix \ref{appendix_validation_approximation} for the results.

Accordingly, the average probability of success for the entanglement distillation schemes is given by:
\begin{equation}\label{eq_Prob_dis_e2e}\small
        P_{\mathrm{dis}}(\theta_i,x_i,y_i,z_i) = P_{\mathrm{dis}}(\langle F'_{\mathrm{W}}\rangle (\theta_i,x_i,y_i,z_i),z_i),
\end{equation}
and the average final e2e fidelity of entangled states distributed between the QS and user $i\in\mathcal{U}$ after distillation is given by: 
\begin{equation}\small
    F_{\mathrm{e2e}}(\theta_i,x_i,y_i,z_i) = F_{\mathrm{dis}}(\langle F'_{\mathrm{W}}\rangle (\theta_i,x_i,y_i,z_i),z_i),
\end{equation}
where $\langle F'_{\mathrm{W}}\rangle (\theta_i,x_i,y_i,z_i)$, derived in \eqref{eq_average_fidelity_before_distillation}, and is based on the average decoherence noise parameter averaged over all possible SIS, developed in Theorem \ref{theorem_avg_noise}.

\vspace{-0.3cm}
\section{Delay-Minimizing QS Resource Allocation Optimization Formulation}\label{sec_optimization}

Now that we have an expression for e2e average fidelity for an entangled state distributed between the QS and a user, we will use it to ensure minimum average fidelity requirements for the different users are satisfied. Moreover, users require different average entanglement generation rates with the smallest possible e2e delays during the almost-infinite episode of QS operation. All these goals can be achieved by carefully optimizing the allocation of QS resources, which include SPSs, quantum memories, and entanglement distillation protocols. In this section, we derive expressions for the average e2e entanglement distribution delay and average entanglement generation rate. Then, we formulate a QS resource allocation optimization problem to fairly minimize the e2e entanglement distribution delay for the different users while satisfying their minimum requirements on average e2e rate and fidelity.

\vspace{-0.31cm}
\subsection{Average Entanglement Distribution Delay}\vspace{-0.14cm}
In this subsection we focus on the problem of minimizing the total average delay associated with the successful generation of an e2e entangled pair of qubits between the QS and user $i\in\mathcal{U}$, subject to minimum required average rate and fidelity. This e2e delay must account for the time taken to perform all quantum operations, including entangled photon emission, its transmission, processing, and memory storage, as well as entanglement distillation delay. In particular, the QS and each user $i\in\mathcal{U}$ focus on entanglement generation until $z_i$ entangled photons are emitted and transferred, where each photon travels a distance $d'_i$ from the QS (or $d_i-d'_i$ from user side) towards the heralding station. There, quantum operations and BSMs are performed, where photons are detected and measured to establish spin-spin entanglement between the QS and corresponding user. 

Denote the time taken to perform every spin-photon entanglement generation attempt as $T_{\mathrm{att}}$. Moreover, let $T_{\mathrm{transfer},i} = \max(\frac{2d'_i}{c_\mathrm{l}},\frac{2(d_i-d'_i)}{c_\mathrm{l}})$ be the time needed to transfer a single entangled photon from the QS or user towards the heralding station (based on which one travels farther). Here, $c_\mathrm{l}$ is the speed of light, and the factor of $2$ is to account for the classical communication time needed to herald the BSM results. Since these operations are probabilistic, with a probability of success given by \eqref{eq_prob_succ_single_spin_spin_entanglement}, the average delay associated with the initial successful generation of $z_i$ spin-spin entangled links between the QS and user $i\in\mathcal{U}$ is $T_{\mathrm{in}}(\theta_i,x_i,z_i) = \frac{z_i \left(T_{\mathrm{transfer},i} +  T_{\mathrm{att}}\right)}{P_{\mathrm{in}}(\theta_i,x_i)}$. Note that we mainly focus on quantum operations that have an associated delay in the order of magnitude of the single-photon transfer time.

A $z_i$-based entanglement distillation circuit is applied to every successfully generated set of $z_i$ spin-spin entangled qubits between the QS and user $i\in\mathcal{U}$. This consumes time $T_{\mathrm{dis}}(y_i,z_i)$ and has a probability of success $P_{\mathrm{dis}}(\theta_i,x_i,y_i,z_i)$. Henceforth, the resulting \emph{average e2e delay for distributing an e2e entangled pair} is \vspace{-0.12cm}
\begin{equation}\small
    T_{\mathrm{e2e}}(\theta_i,x_i,y_i,z_i) = \frac{T_{\mathrm{in}}(\theta_i,x_i,z_i) + T_{\mathrm{dis}}(y_i,z_i)}{P_{\mathrm{dis}}(\theta_i,x_i,y_i,z_i)}.\vspace{-0.1cm}
\end{equation}

Accordingly, the expected average e2e entanglement generation rate between the QS and user $i\in\mathcal{U}$ is  \vspace{-0.13cm}
\begin{equation}\small
    R_{\mathrm{e2e}}(\theta_i,x_i,y_i,z_i) = \frac{1}{T_{\mathrm{e2e}}(\theta_i,x_i,y_i,z_i)}.\vspace{-0.1cm}
\end{equation}

\vspace{-0.36cm}
\subsection{QS Resource Allocation Optimization Problem}\vspace{-0.11cm}
Now, we formulate a QS resource allocation problem to fairly minimize the average delay in distributing an e2e entangled pair experienced by the different users subject to constraints on the average e2e entanglement generation rate, $R_{\mathrm{min},i}$, and fidelity $F_{\mathrm{min},i},\forall i \in \mathcal{U}$. To ensure fairness between the QN users, we consider a \emph{min-max problem}, where we allocate the QS resources so as to minimize the average e2e entangled pair distribution delay experienced by all users. 

In doing so, the considered QS resources to serve users $i\in\mathcal{U}$, where $\abs{\mathcal{U}}=M$, or the controllable optimization variables in our proposed optimization formulation are: 1) $\boldsymbol{\theta} = [\theta_1,\dots,\theta_i,\dots,\theta_M]$, the entanglement generation parameter for each user, which enables the entanglement generation rate-fidelity tradeoff, 2) $\boldsymbol{x} = [x_1,\dots,x_i,\dots,x_M]$, the NV center type allocated to each user, 3) $\boldsymbol{y} = [y_1,\dots,y_i,\dots,y_M]$, the nuclear spin region around the NV center's electron spin selected for each user, and 4) $\boldsymbol{z} = [z_1,\dots,z_i,\dots,z_M]$, the entanglement distillation scheme chosen for each user, which corresponds to the number of e2e entangled links undergoing distillation. Accordingly, the QS resource allocation optimization problem can be formulated as follows:

\begin{subequations}\small
\setlength{\abovedisplayskip}{-3mm}
\begin{alignat}{2}
\mathcal{P}1: \quad &\!\min_{\boldsymbol{\theta},\boldsymbol{x},\boldsymbol{y},\boldsymbol{z}} \max&\quad& T_{\mathrm{e2e}}(\theta_i,x_i,y_i,z_i),\quad\forall i\in\mathcal{U},\label{eq:optProb}\\
&     \mathrm{s. t.}             &      & R_{\mathrm{e2e}}(\theta_i,x_i,y_i,z_i)\geq R_{\mathrm{min},i}, \forall i\in\mathcal{U},\label{eq:constraint1}\\
&                  &      & F_{\mathrm{e2e}}(\theta_i,x_i,y_i,z_i) \geq F_{\mathrm{min},i}, \forall i\in \mathcal{U}, \label{eq:constraint2}\\
&                  &      & \sum_{i\in\mathcal{U}} {\mathbbm{1}}_{\{x_i = 0\}}\leq M_1, \quad\forall i\in\mathcal{U},\label{eq:constraint3}\\
&                  &      & \sum_{i\in\mathcal{U}} {\mathbbm{1}}_{\{x_i = 1\}}\leq M_2, \quad\forall i\in\mathcal{U},\label{eq:constraint4}\\
&                  &      & x_i, y_{i}\in\{0,1\}, \quad\forall i\in\mathcal{U},\label{eq:constraint5}\\
&                  &      & z_{i}\in\{1, 2,..., 7\}, \quad\forall i\in\mathcal{U},\label{eq:constraint6}\\
&                  &      & \theta_{i}\in(0, \frac{\pi}{2}), \quad\forall i\in\mathcal{U},\label{eq:constraint7}\vspace{-0.2cm}
\end{alignat}
\end{subequations}
where the first two constraints, \eqref{eq:constraint1} and \eqref{eq:constraint2}, ensure satisfying each user's minimum required average e2e entanglement generation rate and their average e2e fidelity, respectively. Moreover, \eqref{eq:constraint3} and \eqref{eq:constraint4} represent the QS capacity constraints that ensure that the number of allocated NV centers from type 1 and 2 do not exceed the available QS SPSs. Finally, \eqref{eq:constraint5}, \eqref{eq:constraint6}, and \eqref{eq:constraint7} specify the possible range of values for each of the optimization control variables.

\vspace{-0.4cm}
\subsection{Proposed Solution}\label{sec_solutions}\vspace{-0.1cm}
The proposed problem $\mathcal{P}1$ is a non-convex optimization problem, that is generally an \emph{NP}-hard problem. Here, we note that the derivatives of the functions in \eqref{eq:optProb}, \eqref{eq:constraint1}, and \eqref{eq:constraint2} are computationally expensive and not tractable. Accordingly, to solve $\mathcal{P}1$, we apply a simulated annealing algorithm, which is a metaheuristic, derivative-free technique, as shown in Algorithm \ref{simulated_annealing_algorithm}.

Algorithm \ref{simulated_annealing_algorithm} starts by randomly picking an initial feasible solution, and an initially high temperature. The initial solution must satisfy the different constraints of $\mathcal{P}1$, and it is set to be the initial best solution. In every iteration of the simulated annealing algorithm, for every user $i\in\mathcal{U}$, we generate a proximal, or neighbor, solution for user $i$, which represents a minor random adjustment to the current solution for that user within the feasibility bounds. Then, we calculate the resulting difference in the objective function which is the total average e2e entanglement distribution delay experienced by user $i\in\mathcal{U}$, denoted by $\Delta T_i$. If the neighbor solution for user $i\in\mathcal{U}$ results in a lower average delay for that user, it is accepted as the current solution. Moreover, it could also be accepted if it satisfies the Metropolis criterion (see line (11) in Algorithm \ref{simulated_annealing_algorithm}), which helps avoiding local minima solutions, otherwise the current solution for user $i\in\mathcal{U}$ remains unchanged. 

Then, the resulting largest average delay among the different users under their corresponding new solutions is compared to their previous largest average delay under the current best solution. If the new maximum average delay is smaller, then the new solution becomes the best solution. This process is repeated while reducing the system temperature using an exponential cooling schedule, until reaching a minimum temperature threshold. Finally, as the algorithm finishes, the resulting approximate best solutions $(\boldsymbol{\theta}^*, \boldsymbol{x}^*, \boldsymbol{y}^*, \boldsymbol{z}^*)$, that are near-optimal are returned \cite{kirkpatrick1983optimization}.  

\setlength{\textfloatsep}{0pt}
\begin{algorithm}[t]\small
\caption{Simulated Annealing Algorithm for $\mathcal{P}1$}\footnotesize\label{simulated_annealing_algorithm}
\begin{algorithmic}[1]
\State Initialize the current solution $\boldsymbol{\theta}$, $\boldsymbol{x}$, $\boldsymbol{y}$, and $\boldsymbol{z}$ to a random feasible solution in the search space
\State Set initial temperature $\tau_{\mathrm{sol}}=\tau_0$
\State Define $K$, the number of performed iterations for each temperature level
\State Initialize the optimal solution $\boldsymbol{\theta}^*$, $\boldsymbol{x}^*$, $\boldsymbol{y}^*$, and $\boldsymbol{z}^*$ to the current solution
\While{$\tau_{\mathrm{sol}}>\tau_{\min}$}
    \For{$k=1$ to $K$}
        \For{each user $i \in \mathcal{U}$}
            \State Generate random neighbors $\theta'_i$, $x'_i$, $y'_i$, and $z'_i$ for user $i$ by slightly altering $\theta_i$, $x_i$, $y_i$, and $z_i$, respectively
             \If{all generated neighbors $\theta'_{i}$, $x'_{i}$, $y'_{i}$, and $z'_{i},$ $\forall i \in \mathcal{U}$, satisfy $\mathcal{P}1$ constraints \eqref{eq:constraint1}, \eqref{eq:constraint2}, \eqref{eq:constraint3}, \eqref{eq:constraint4}, \eqref{eq:constraint5}, \eqref{eq:constraint6}, and \eqref{eq:constraint7}}
            \State Calculate $\Delta T_i = T_{\mathrm{e2e}}(\theta'_i,x'_i,y'_i,z'_i) - T_{\mathrm{e2e}}(\theta_i,x_i,y_i,z_i)$
                \If{$\Delta T_i < 0$ \textbf{or} $\exp(-\Delta T_i / \tau_{\mathrm{sol}}) > r\in\mathrm{Uniform}[0,1]$}
                    \State Update $\theta_i$, $x_i$, $y_i$, and $z_i$ to $\theta'_i$, $x'_i$, $y'_i$, and $z'_i$, respectively
                \EndIf
            \EndIf
        \EndFor
    \EndFor
    \State Find $m = \arg\max_{i \in \mathcal{U}} T_{\mathrm{e2e}}(\theta_i,x_i,y_i,z_i)$
    \State Calculate global $\Delta T = T_{\mathrm{e2e}}(\theta_{m},x_{m},y_{m},z_{m}) - \max_{i \in \mathcal{U}} T_{\mathrm{e2e}}(\boldsymbol{\theta}^*,\boldsymbol{x}^*,\boldsymbol{y}^*,\boldsymbol{z}^*)$
    \If{$\Delta T < 0$}
        \State Update the best solution $\boldsymbol{\theta}^*$, $\boldsymbol{x}^*$, $\boldsymbol{y}^*$, and $\boldsymbol{z}^*$ to the current solution
    \EndIf
    \State Update $\tau_{\mathrm{sol}}$ according to adopted cooling schedule, e.g., exponential cooling, as $\tau_{\mathrm{sol}} = \tau_{\mathrm{sol}} \cdot \alpha_{\mathrm{sol}}$ 
\EndWhile
\State \textbf{return} the best solution found, $(\boldsymbol{\theta}^*, \boldsymbol{x}^*, \boldsymbol{y}^*, \boldsymbol{z}^*)$
\end{algorithmic}
\end{algorithm}
\vspace{-0.2cm}
\section{Simulations and Results}\label{sec_simulations}\vspace{-0.05cm}
Now, we conduct exhaustive simulations to evaluate the efficiency and assess the performance of our proposed QS resource allocation framework and how it compares to state-of-the-art existing frameworks. The following \emph{default setup} of the QN is adopted throughout the conducted simulations: 1) users in the QN are randomly located around the QS with their distances from the QS being randomly sampled from a uniform distribution, where $d_i \sim \mathrm{Uniform}(100, 1500)$\,m, $\forall i \in \mathcal{U}$, 2) unless specified otherwise, the users' minimum required average entanglement generation rate, $R_{\mathrm{min},i}, \forall i\in\mathcal{U}$, is assumed to be sampled from a uniform distribution such that $R_{\mathrm{min},i} \sim \mathrm{Uniform}[1,1000]$\,Hz, 3) users' corresponding minimum required average e2e fidelity, $F_{\mathrm{min},i}\forall i \in\mathcal{U}$, is also sampled from a uniform distribution such that $F_{\mathrm{min},i} \sim \mathrm{Uniform}[0.85,0.99]$, which accounts for the requirements of different quantum applications \cite{chehimi2024reconfigurable}, and 4) we assume low phase uncertainty of transmitted entangled photons, and we consider nuclear spin coherence time in the proximity of one minute \cite{bradley2019ten,bartling2022entanglement}, which vary based on the isotopic decomposition of NV centers and the coupling strength between its nuclear and electron spins \cite{taminiau2012detection,doherty2013nitrogen,brandenburg2018improving}. Note here that since we do not consider a large-scale QN and we assume user to be within $1.5$\,km from the QS, we can consider the utilization of NV center-emitted entangled photons at their optical frequency, without applying frequency conversion techniques \cite{rozpkedek2019near}. The rest of the different QN parameters adopted in the default setup are provided in Table \ref{tbl:Def}. Unless specified otherwise, presented results are statistically averaged averaged over 1,000 Monte Carlo runs, and parameters of the default setup are considered.

\begin{table}[t!] 
\caption{Summary of Simulation Parameters.}\vspace{-0.2cm}
\label{tbl:Def}\footnotesize
\centering
\begin{tabular}{|p{1.3cm}|p{4cm}|p{2cm}|}
 \hline
 {\bf Parameter} & {\bf Description} & {\bf Value} \\ \hline
 $P_\mathrm{e}$ & probability of entangled photon emission into ZPL & $0.46-0.6$ \cite{hensen2015loophole,orphal2023optically,ruf2019optically}\\ \hline
 $P_\mathrm{ce}$ & probability of single-photon collection & $0.49-0.8$ \cite{bogdanovic2017design,orphal2023optically}\\ \hline
 $P_\mathrm{det}$ & probability of detecting photons at heralding station, conditioned on being emitted & $0.85$ \cite{bhaskar2020experimental}\\ \hline
 $T_{\mathrm{prep}}$ & time for a spin-spin entanglement generation attempt & $5.5\,\mu$s \cite{humphreys2018deterministic}\\ \hline 
 $\lambda_{\mathrm{prep}}$ & dephasing noise parameter for preparing a spin-photon entangled state & $0.99$ \cite{hensen2015loophole}\\ \hline 
$L_0$ & attenuation Length - visible light frequency & $0.542$\,km \cite{hensen2015loophole}\\ \hline 
$\alpha_{\mathrm{gate}}$ & 2-qubit gate fidelity & $0.97-0.999$ \cite{takou2023precise}\\ \hline
$T_{\mathrm{gate}}$ & 2-qubit gate time & $68-400\,\mu$s \cite{takou2023precise}\\ \hline 
\end{tabular} 
\end{table}

In order to benchmark the performance of our proposed framework, we compare its achieved performance with the following frameworks: 1) \emph{Delay-agnostic (DA) framework}, which represents existing entanglement distribution frameworks that optimize QS resources, without considering the QS capable of performing entanglement distillation, i.e., $z_i=1, \forall i \in\mathcal{U}$, 2) \emph{Physics-agnostic (PA) framework}, which represents the majority of existing frameworks that, unlike our proposed framework, do not consider controlling practical physics-based NV center characteristics, e.g., isotopic decomposition and nuclear spin region, and 3) \emph{Minimum distillation (MD) framework}, which considers always performing the minimum amount of distillation, that is distilling two entangled pairs to yield one entangled pair with higher fidelity, i.e., $z_i=2, \forall i \in\mathcal{U}$. Here, we apply PA frameworks by randomly choosing $x_i$ and $y_i, \forall i \in\mathcal{U}$, without being optimized.

We start by considering a QN in the default setup consisting of a QS with $M = 1$, $M_1 = 2$, and $M_2 = 2$, serving $4$ users, ordered based on their average distances from the QS, i.e., user $1$ is closest to the QS, and user $4$ is farthest. We show first, in Fig. \ref{fig_result_T_e2e}, the average achieved e2e entanglement distribution delay for the different users. We observe from Fig. \ref{fig_result_T_e2e} that our proposed framework is the only one that manages to satisfy the different user requirements, while the other frameworks fail to do so. Particularly, the DA and MD frameworks fail to satisfy the minimum fidelity requirement of users 3 and 4, and the PA framework fails to satisfy the requirements of user 4, which resembles the farthest user from the QS, on average. Particularly, while sending entangled photons directly without distillation, as in the DA framework, or performing minimal distillation, as in the MD framework, result in low average e2e delay, it is necessary, in many cases, to perform entanglement distillation protocols with more distilled qubits in order to meet the minimum required average e2e fidelity by the users. While such large distillation protocols result in increased delay, they allow the QS to successfully satisfy all user requirements, as achieved by our proposed framework. Moreover, we observe, from Fig. \ref{fig_result_T_e2e}, that it is optimal to send entangled photons directly to users 1 and 2 without performing entanglement distillation. This is mainly due to their average lower experienced channel losses due to their closeness to the QS. In serving such users, performing entanglement distillation results in increased, unnecessary, entanglement distribution delay due to the time consumed in storing entangled qubits for distillation and the time needed to run the distillation circuit. Additionally, we observe from Fig. \ref{fig_result_T_e2e} that our proposed framework results in more than $50\%$ reduction in the average e2e delay, $T_{\mathrm{e2e}}$, needed for entanglement distribution, compared to the MD framework. Furthermore, Fig. \ref{fig_result_T_e2e} shows that our proposed framework manages to achieve, on average, around $30\%$ reduction in $T_{\mathrm{e2e}}$ compared to the PA framework. This clearly shows the significant importance of our proposed control of the isotopic decomposition of NV centers and selection of nuclear spin regions, specified by the variables $x_i$ and $y_i, \forall i \in\mathcal{U}$. Moreover, we verify, through comparison with exhaustive search solution, that the obtained results from our deployed simulated annealing algorithm are almost optimal, within 1\% of the optimal results, while saving, on average, around $55\%$ of the running time. Additionally, it is noteworthy to mention that both the DA and the MD frameworks are special cases of our proposed framework. For instance, our framework chooses $z_i=1$, on average, for the first two users in Fig. \ref{fig_result_T_e2e}, which aligns with the DA framework.

\begin{figure}
\centering
    \includegraphics[scale=0.19]{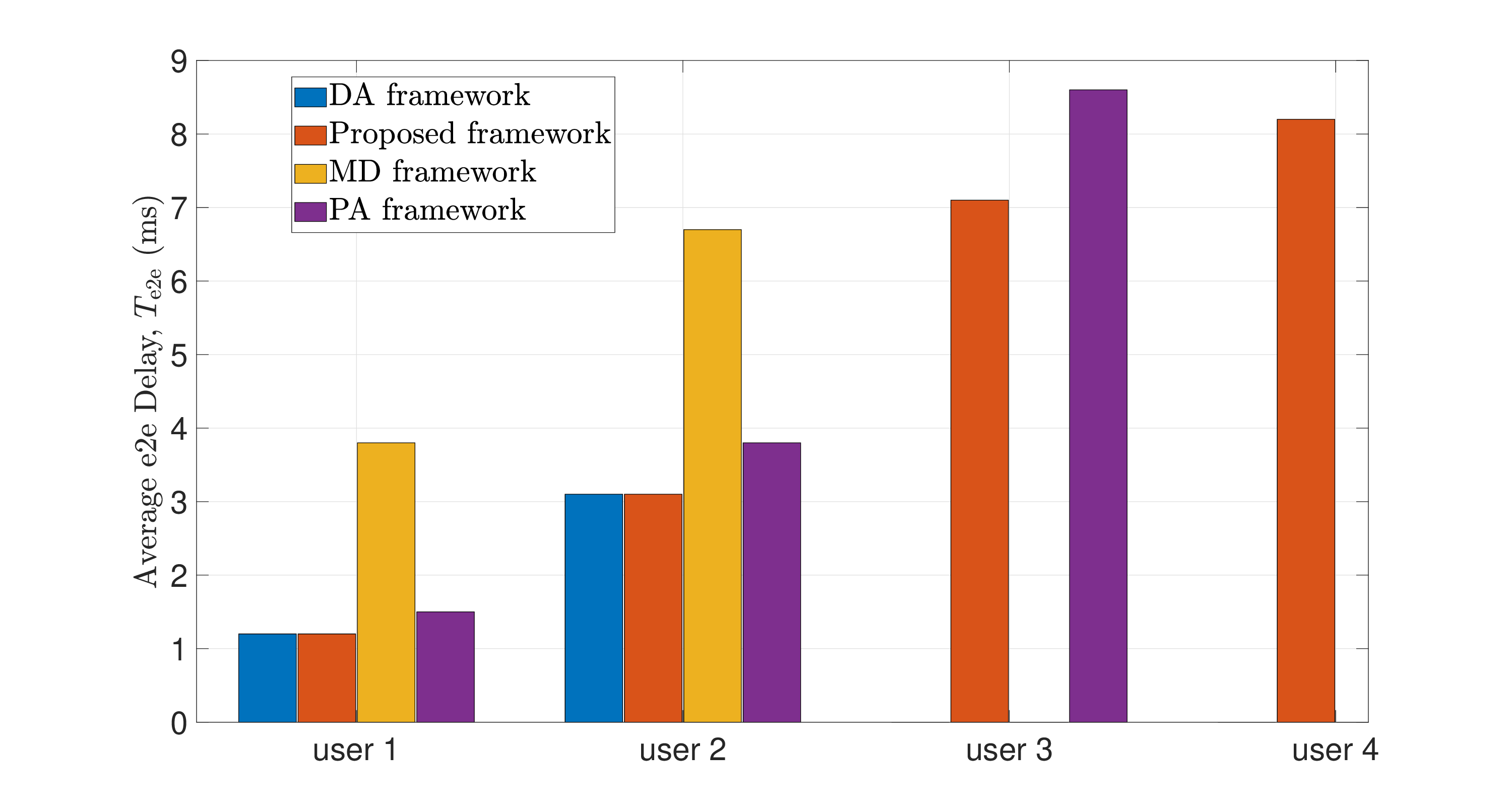}\vspace{-0.4cm}
    \caption{Average achieved e2e entanglement distribution delay, $T_{\mathrm{e2e}}$, for the different QN users.}
    \label{fig_result_T_e2e}
\end{figure}

Next, in Fig. \ref{fig_result_F_e2e}, we analyze the achieved e2e average fidelity for the QN users by the different considered benchmarks. In particular, we consider a similar QN setup as described for Fig. \ref{fig_result_T_e2e}. Here, we fix the minimum required average e2e fidelity for the different users to be $0.85$, $F_{\mathrm{min},i} = 0.85, \forall i\in\mathcal{U}$, and the minimum required average e2e rate for all users to 100\,Hz, i.e., $R_{\mathrm{min},i} = 10$\,Hz, $\forall i\in\mathcal{U}$. We observe, from Fig. \ref{fig_result_F_e2e}, that our proposed framework is the only one that manages, on average, to satisfy all user requirements in the considered QN. In particular, all considered frameworks manage to satisfy the minimum average e2e fidelity requirements, $F_{\mathrm{min},i}$ for users that are, on average, close to the QS. However, as distance increases, the DA framework fails to satisfy the minimum average e2e fidelity requirements, and we observe around $5\%$ reduction in the achieved $F_{\mathrm{e2e}}$ by the DA framework compared to our framework, which performs entanglement distillation. However, if the selection of the applied distillation protocol is not optimized, i.e., $z_i$ is fixed $\forall i\in\mathcal{U}$, as is the case with the MD framework, we observe up to $11\%$ reduction in average $F_{\mathrm{e2e}}$ achieved by the MD framework compared to our proposed framework. Furthermore, if the physics characteristics of NV centers are not controlled and optimized, i.e., $x_i$ and $y_i$ are randomly selected for every user $i\in\mathcal{U}$, as is the case with the PA framework, we observe around $7\%$ reduction in $F_{\mathrm{e2e}}$ compared to our proposed framework.

\begin{figure}
    \centering
 \includegraphics[scale=0.19]{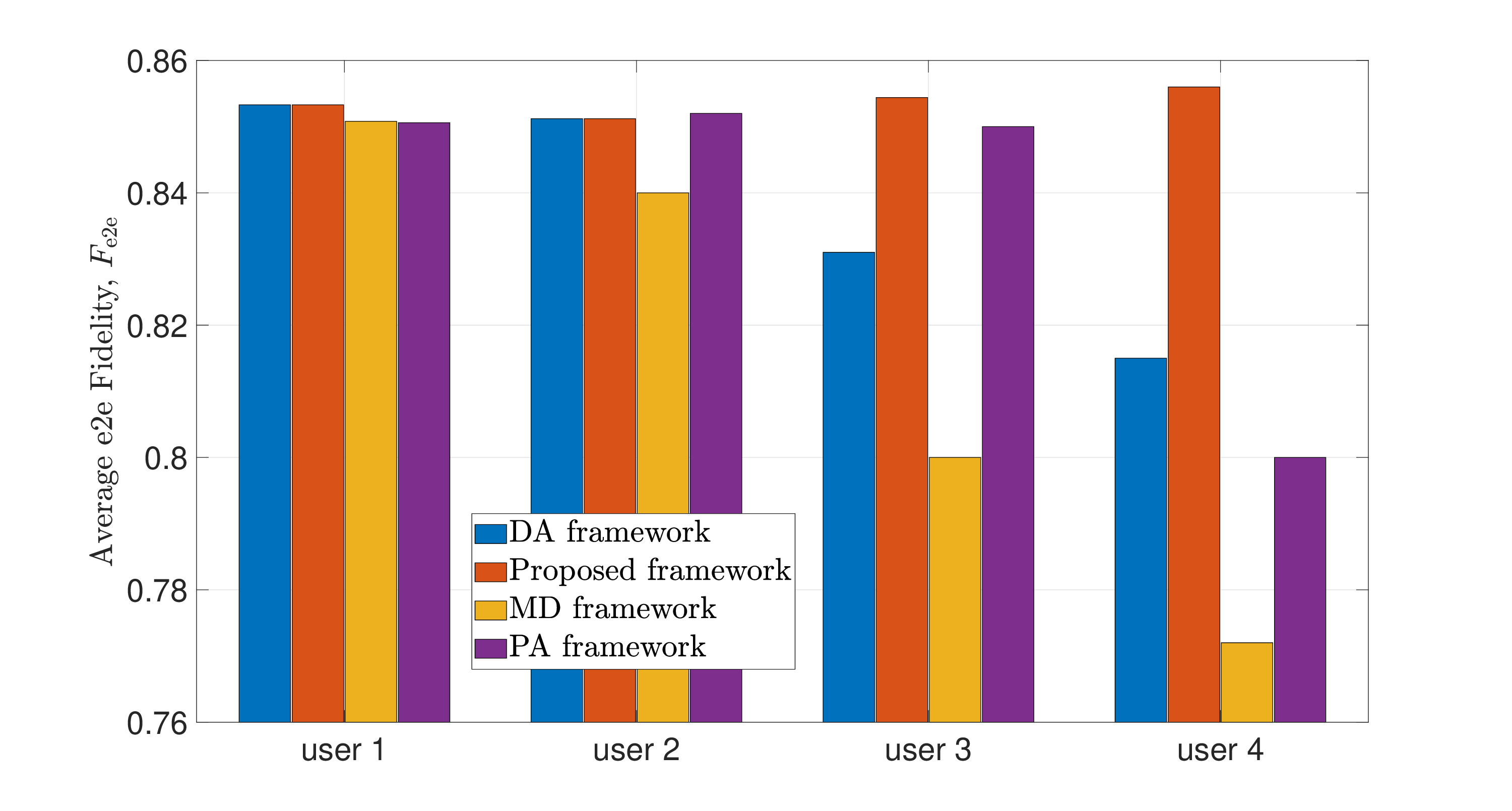}\vspace{-0.45cm}
    \caption{Average achieved e2e fidelity, $F_{\mathrm{e2e}}$, for the different QN users.}\vspace{-0.5cm}
    \label{fig_result_F_e2e}
\end{figure}

Another interesting observation from Fig. \ref{fig_result_F_e2e} is that the MD framework resulted in a lower $F_{\mathrm{e2e}}$ for the last two users compared to the DA framework. This may appear unexpected, as distillation is anticipated to yield improved fidelity, particularly when the input fidelity exceeds $0.5$. However, this is not the case, since the input fidelity to the $z_i = 2$ distillation protocol in the MD framework is not equal to the average e2e fidelity obtained through the DA framework. To see this, one must pay attention to the entanglement generation parameter for each user, i.e., $\boldsymbol{\theta}$, which has a significant impact on the resulting average e2e fidelity, and must be chosen carefully. In particular, the value of $\boldsymbol{\theta}$ chosen for the DA framework would be significantly different from the value chosen for the MD framework. Next, we delve into the impact of $\boldsymbol{\theta}$ and $\boldsymbol{y}$ on the resulting average e2e fidelity.

\begin{figure}
    \centering
 \includegraphics[scale=0.19]{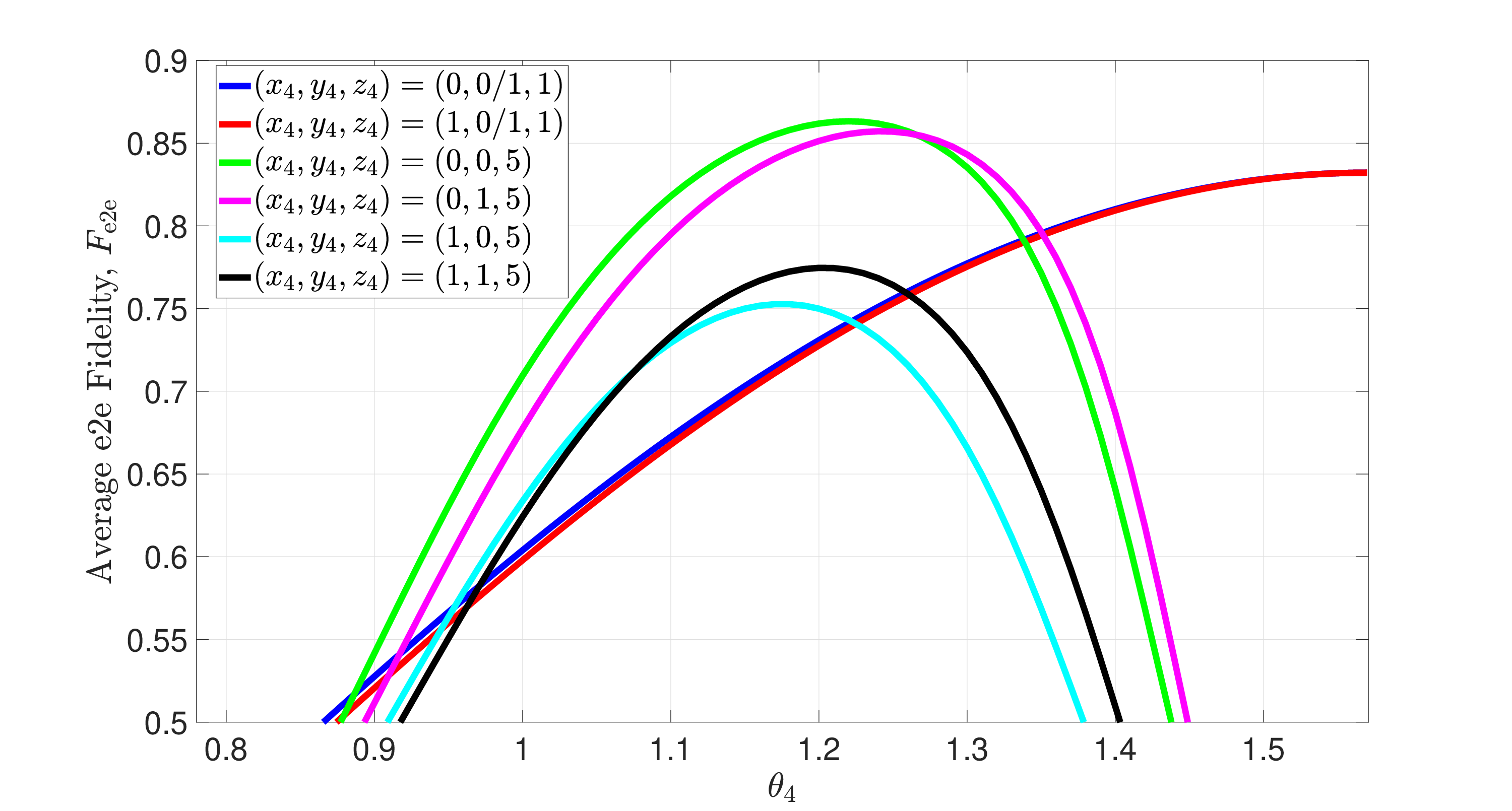}\vspace{-0.4cm}
    \caption{Average achieved e2e fidelity, $F_{\mathrm{e2e}}$, for a user $i=4$ with an average distance of $d_4 = 1$\,km from the QS vs $\theta_4$, in a 4-user QN.}
    \label{figure_result_F4_vs_theta}
\end{figure}

In Fig. \ref{figure_result_F4_vs_theta}, we focus on a single user, denoted user 4, from the previous 4-user QN setup considered for Fig. \ref{fig_result_F_e2e}. This user is relatively far away from the QS, and we analyze its achieved average e2e fidelity as its assigned $\theta_4$, $x_4$, $y_4,$ and $z_4$ values are varied. In particular, the optimal distillation protocol for user 4 is found to be $z_4 = 5$ using our proposed framework. Thus, we compare, in Fig. \ref{figure_result_F4_vs_theta}, our proposed framework and the DA framework, where $z_4 = 1$. Note that when $z_4 = 1$, no distillation is performed and thus no nuclear spin is used as a quantum memory because qubits are directly shared without storage, thus, the selection of $y_4$ becomes irrelevant in this case. We observe from Fig. \ref{figure_result_F4_vs_theta} that the DA framework can achieve a high $F_{\mathrm{e2e}}$ when $\theta_4$ takes very high values, i.e., near $\pi/2$, which is not the case when distillation is performed, since all cases with $z_4=5$ achieve a very low $F_{\mathrm{e2e}}$ at such high values of $\theta_4$ in Fig. \ref{figure_result_F4_vs_theta}. The main reason behind this behavior is the special role for $\theta_i$ in the fidelity expression, based on \eqref{eq_F_in_prime}, and the probability of success expression in \eqref{eq_prob_succ_single_spin_spin_entanglement}, for a user $i\in\mathcal{U}$. In particular, when $\theta_i$ is very high, $\sin^2(\theta)\sim 1$ and $\cos^2(\theta) \sim 0$. As such, $P_{\mathrm{in}}\sim 0$ and $\langle F'_{\mathrm{in}}\rangle \sim a\times \langle \lambda_{\mathrm{d,avg}}(x_i,y_i,z_i) \rangle + \frac{(1-\langle \lambda_{\mathrm{d,avg}} \rangle(x_i,y_i,z_i))}{4}$. But, when $z_i=1$, $\langle \lambda_{\mathrm{d,avg}} \rangle(x_i,y_i,z_i) = 1$ since there is no memory decoherence. As such, the resulting fidelity when no distillation is performed becomes significantly high at high values of $\theta_i, \forall i\in\mathcal{U}$. In contrast, when $z_i=5$, the impact of memory decoherence becomes paramount at high values of $\theta_i$.

Moreover, we observe from Fig. \ref{figure_result_F4_vs_theta} that the NV center type, or isotopic decomposition, represented by $x_4$, has a marginal effect on $F_{\mathrm{e2e}}$ when no distillation is performed ($z_4=1$). In contrast, when distillation is performed, e.g., $z_4=5$, we observe from Fig. \ref{figure_result_F4_vs_theta} that both $x_4$ and $y_4$ has a significant impact on $F_{\mathrm{e2e}}$. Note that for the considered scenario, performing distillation with type-2 NV centers, i.e., $z_4=5$ and $x_4=1$, does not result in satisfying the minimum fidelity requirement of $0.85$. However, we can observe the significant impact of changing $y_4$ from $0$ to $1$ on boosting the obtained e2e average fidelity over a wide range of possible $\theta_4$ values due to the resulting enhancement in the coherence time. However, since type-1 NV centers already have a high coherence time, the impact of varying $y_4$ on $F_{\mathrm{e2e}}$ is marginal.

Finally, for the same previous QN setup with user 4, we show, in Fig. \ref{fig_result_R4_vs_theta}, the resulting average e2e entanglement generation rate, $R_{\mathrm{e2e}}$, as its assigned entanglement generation parameter $\theta_4$ varies for different combinations of $x_4,$ $y_4,$ and $z_4$. Fig. \ref{fig_result_R4_vs_theta} clearly shows that the enhanced fidelity achieved through distillation comes at the heavy price of a significantly reduced entanglement generation rate. Specifically, the most basic distillation scheme, when $z_4=2$, results, in a perfect best case scenario where $P_{\mathrm{dis}} = 1$, in $50\%$ reduction in $R_{\mathrm{e2e}}$. This justifies the major drop in $R_{\mathrm{e2e}}$ observed in Fig. \ref{fig_result_R4_vs_theta} when $z_4 = 5$. Furthermore, while $x_4$ has a marginal effect on $F_{\mathrm{e2e}}$ when no distillation is performed ($z_4=1$) in Fig. \ref{figure_result_F4_vs_theta}, we observe from Fig. \ref{fig_result_R4_vs_theta} that $x_i$ has a significant effect on the achieved $R_{\mathrm{e2e}}$ for user $i\in\mathbf{U}$. This is because the impact of $x_i$ appears mainly in the probability of success expression for generating a spin-spin entangled link between the QS and user $i\in\mathcal{U}$. We can see a similar impact for $x_i$ when we perform distillation, but over a smaller range of $\theta_i$ values. Finally, we observe from Fig. \ref{fig_result_R4_vs_theta} that the selection of $y_i$ has a minimal effect on $R_{\mathrm{e2e}}, \forall i \in \mathcal{U}$.  

\begin{figure}
    \centering\vspace{-0.1cm}
 \includegraphics[scale=0.19]{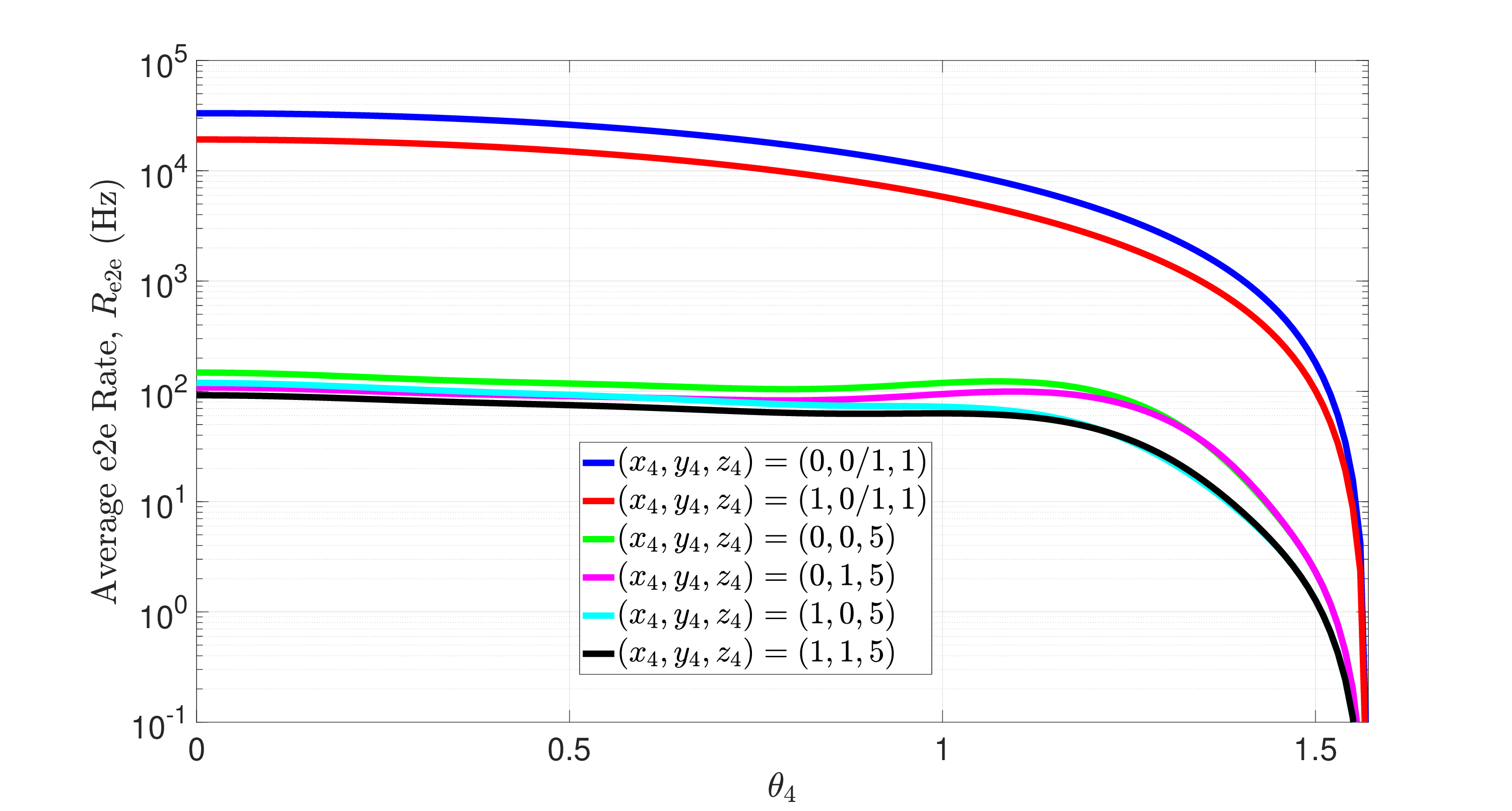}\vspace{-0.4cm}
    \caption{Average achieved e2e entanglement generation rate, $R_{\mathrm{e2e}}$, for a user $4$ with an average distance of $d_4 = 1$\,km from the QS vs $\theta_4$, in a 4-user QN.}
    \label{fig_result_R4_vs_theta}
\end{figure}







\section{Conclusion}\label{sec_conclusion}
In this paper, we have proposed a QS resource allocation framework that jointly optimizes the entanglement distribution delay and the entanglement distillation operations while satisfying heterogeneous user requirements on average e2e rate and fidelity. We consider the control of NV center characteristics like its isotopic decomposition and nuclear spin region to enhance the QS performance. Moreover, we derive the analytical expressions for average memory decoherence noise. The proposed framework is shown, through extensive simulations, to result in around $30\%$ to $50\%$ reductions in the average e2e entanglement distribution delay and around $5\%$, $7\%$, and $11\%$ reductions in the average e2e fidelity compared to existing frameworks.

\appendices
\vspace{-0.4cm}
\section{Derivation of Average Noise in Theorem \ref{theorem_avg_noise}}\label{appendix_average_noise_derivation}\vspace{-0.15cm}
We will prove here the formula for the expectation value of the QS average memory decoherence depolarizing noise parameter $\langle \lambda_{\text{d,avg}} \rangle(x_i,y_i,z_i)$, shown in \eqref{eq_average_decoherence_noise_parameter}. This noise is averaged over all possible SISs, and is considered to be uniform over $z_i$ stored entangled qubit pairs by applying a twirling map. These qubits are stored in nuclear spins in region $y_i$ in a type-$x_i$ NV center in diamond for serving user $i\in\mathcal{U}$.

For conciseness, we will set $\langle\lambda\rangle = \langle \lambda_{\text{d,avg}} \rangle(x_i,y_i,z_i)$, $R \equiv e^{-\frac{1}{T_{\text{c}}(x_i,y_i)}}$, $z\equiv z_i$, and $p\equiv P_{\mathrm{in}}(\theta_i,x_i), \forall i\in\mathcal{U}$. Here, $P_{\mathrm{in}}(\theta_i,x_i)$ is the probability of successfully establishing an e2e spin-spin entangled link between the QS and user $i\in\mathcal{U}$ at distance $d_i$ from the QS. Moreover, $T_{\text{c}}(x_i,y_i)$ represents the quantum memory coherence time for user $i\in\mathcal{U}$ based on its selected NV center type and nuclear spin region.

A specific SIS instance includes $z$ numbers, each of which corresponds to the round number each of the $z$ entangled pairs arrived at. Particularly, each instance is represented by an SIS, $\boldsymbol{S}$, of $z$ non-negative natural numbers $\left(t_1, t_2, \ldots, t_{z}\right)$. We need two ingredients to calculate the desired expectation value; the probability of observing a given sequence $\boldsymbol{S}$ and the associated average noise parameter. First, for the probability, note that for a given sequence $\boldsymbol{S}$ with $\max(\boldsymbol{S}) = t_{z}$ we have $z$ successes and $t_{z}-z$ number of failures, such that the probability is given by $\left(1-p\right)^{t_{z}-z}p^{z}$. Now, for the associated average noise parameter, note that the $k$'th successfully created pair with $1 \leq k \leq z$ has to wait until the last pair is successfully created. That is, the $k$'th pair has to wait for a time $t_{z}-t_{k}$. Since each round adds a multiplicative factor of $R$ to the noise parameter, we find that the average noise parameter for an SIS $\boldsymbol{S}$ is given by $\frac{1}{z}\sum_{k=1}^{z}R^{t_{z}-t_k}$.

We have then that the expectation value of the average noise parameter over all SISs is given by 
\begin{equation}\footnotesize\label{eq:avg_noise_derivation}
\langle\lambda\rangle = \frac{1}{z}\left(\frac{p}{1-p}\right)^z\sum_{\boldsymbol{S}} \left(1-p\right)^{s_{z}}\left(\sum_{k=1}^{z}R^{t_{z}-t_k}\right),
\end{equation}
where the sum is over all SISs. Let us write $q \equiv 1-p$ and focus on the term in the sum,
\begin{equation}\footnotesize
    \sum_{\boldsymbol{S}}q^{t_{z}}\left(\sum_{k=1}^{z}R^{t_{z}-t_k}\right)=\sum_{T=1}^{\infty} q^{T} \sum_{\mathclap{\substack{\boldsymbol{S}\textrm{ s.t. }\\t_{z}=T}}} ~\left(\sum_{k=1}^{z}R^{T-t_k}\right)\label{eq:biggersum},
\end{equation}
where we have split up the sum into one over all $\boldsymbol{S}$ with $t_z = T$ and replaced $t_z$ by $T$. 

Now, we focus on the terms after $q^T$ in \eqref{eq:biggersum}, including the inner two sums. As a warm up question, we will first consider the number of $\boldsymbol{S}$ such that $t_z = T$. The total number is given by $\binom{T-1}{T-z}=\binom{T-1}{z-1}$ since this is equivalent to asking how many ways there are to choose $T-z$ failures (or $z-1$ successes) out of $T-1$ attempts. Note that we consider $T-1$ and not $T$ attempts because by definition the $T$'th attempt is a success. This allows us to split the term into two parts, as follows
\begin{equation}\footnotesize\label{eq:spliteq}
\begin{split}
   &\sum_{\mathclap{\substack{\boldsymbol{S}\textrm{ s.t. }\\t_{z}=T}}} ~\left(\sum_{k=1}^{z}R^{T-t_k}\right)
   =\sum_{\mathclap{\substack{\boldsymbol{S}\textrm{ s.t. }\\t_{z}=T}}} ~\left(1+\sum_{k=1}^{z-1}R^{T-t_k}\right)\\
            =&~\binom{T-1}{z-1}+R^T\sum_{\mathclap{\substack{\boldsymbol{S}\textrm{ s.t. }\\t_{z}=T}}} ~\left(\sum_{k=1}^{z-1}R^{-t_k}\right)\ , 
\end{split}
\end{equation}
where in the second step we have split the sum into two by separately summing from $k=1$ to $z-1$ and the term corresponding to $k=z$, such that $t_z = T$. In the next line we used the fact that there are $\binom{T-1}{z-1}$ SISs with $s_z = T$. 

Let us focus now on the last term in \eqref{eq:spliteq} that is multiplied by $R^T$. Note that choosing an arbitrary $1\leq m\leq T-1$ to be in an SIS $\boldsymbol{S}$ leaves us with choosing $z-2$ successes out of $T-2$ attempts, such that (independently of $m$) there are $\binom{T-2}{z-2}$ such options. It is thus possible to split this term over all possible values of $m$, as follows
\begin{equation}\footnotesize
\begin{split}
\sum_{\mathclap{\substack{\boldsymbol{S}\textrm{ s.t. }\\t_{z}=T}}} ~\left(\sum_{k=1}^{z-1}R^{-t_k}\right)&= \sum_{m=1}^{T-1}\binom{T-2}{z-2}R^{-m}\\ &=\binom{T-2}{z-2}R^{-T}\frac{R-R^T}{1-R},
\end{split}
\end{equation}
which, when plugged into \eqref{eq:spliteq} gives
\begin{equation}\footnotesize
    \begin{split}
             \sum_{\mathclap{\substack{\boldsymbol{S}\textrm{ s.t. }\\t_{z}=T}}} ~\left(\sum_{k=1}^{z}R^{T -t_k}\right)  =&   ~ \binom{T-1}{z-1}+R^T\sum_{\mathclap{\substack{\boldsymbol{S}\textrm{ s.t. }\\t_{z}=T}}} ~\left(\sum_{k=1}^{z-1}R^{-t_k}\right)\\
            =& ~\binom{T-1}{z-1}+\binom{T-2}{z-2}\frac{R-R^T}{1-R},
    \end{split}
\end{equation}
and inputting this into \eqref{eq:biggersum} yields
\begin{equation}\label{eq:split_eq2}\footnotesize
\begin{split}
 &\sum_{T=1}^{\infty} q^{T} \sum_{\mathclap{\substack{\boldsymbol{S}\textrm{ s.t. }\\t_{z}=T}}} ~\left(\sum_{k=1}^{z}R^{T-t_k}\right) \\ 
 &= \left(\sum_{T=1}^{\infty} q^{T}  ~\binom{T-1}{z-1}\right)+\left(\sum_{T=1}^{\infty} q^{T}\binom{T-2}{z-2}\frac{R-R^T}{1-R}\ \right)\,
 \end{split}
\end{equation}
where the sum is split into two terms, tackled individually. 

Now, consider the first term in \eqref{eq:split_eq2}, while including the $\frac{1}{z}\left(\frac{p}{1-p}\right)^z$ factor that we dropped from \eqref{eq:avg_noise_derivation} in the beginning of the analysis, and change back to $p \equiv 1- q$, we get
\begin{equation}\footnotesize
\begin{split}
&\frac{1}{z}\left(\frac{p}{1-p}\right)^z\left(\sum_{T=1}^{\infty} q^{T}  ~\binom{T-1}{z-1}\right)\\ &=\frac{1}{z}\sum_{T=0}^{\infty}\left[\binom{T-1}{z-1}\left(1-p\right)^{T-z}p^z\right]\\
&=\frac{1}{z}\sum_{T=0}^{\infty}\left[\binom{T-1}{T-z}\left(1-p\right)^{T-z}p^z\right],
\end{split}
\end{equation}
where in the last step we have used that $\binom{n}{k}=\binom{n}{n-k}$. Note that the term between brackets is the probability mass function of a negative binomial distribution with $T-z$ failures, $z$ successes and $T$ total attempts, and success probability $p$. Since we sum over all possible values of $T$ and those exhaust all possibilities, the sum needs to equal $1$, i.e.~
\begin{equation}\footnotesize
    \frac{1}{z}\sum_{T=0}^{\infty}\left[\binom{T-1}{T-z}\left(1-p\right)^{T-z}p^z\right]=\frac{1}{z}\label{eq:firstterm}\ .
\end{equation}
which is not surprising since the last successful attempt never decoheres (since the corresponding state does not need to wait), and so it will contribute $\frac{1}{z}$ to the total average noise.

Now for to the second term in equation \eqref{eq:split_eq2},
\begin{equation}\footnotesize
\begin{split}
    \sum_{T=1}^{\infty} q^{T}\binom{T-2}{z-2}\frac{R-R^T}{1-R}
    &=\frac{R}{1-R}\left[\left(\sum_{T=1}^{\infty} q^{T}\binom{T-2}{z-2}\right)\right]\\ -& \frac{1}{1-R}\left[\left(\sum_{T=1}^{\infty} \left(Rq\right)^{T}\binom{T-2}{z-2}\right)\right] \ .
    \end{split}
\end{equation}

We see that the two terms in between the square brackets are of the form
\begin{equation}\footnotesize
\begin{split}
&\sum_{T=1}^{\infty} \mu^{T}\binom{T-2}{z-2} =  \mu^2\sum_{T=-1}^{\infty} \mu^{T}\binom{T}{z-2} \\
&=  \mu^2\left(\mu^{-1}\binom{-1}{z-2}+\sum_{T=0}^{\infty} \mu^{T}\binom{T}{z-2}\right)\\
&=  \left(-1\right)^{z}\mu+\mu^2\sum_{T=0}^{\infty} \mu^{T}\binom{T}{z-2} = \left(-1\right)^{z}\mu+\frac{\mu^{z}}{\left(1-\mu\right)^{z-1}}\ .
\end{split}
\end{equation}\vspace{-0.1cm}

It follows that
\begin{equation}\footnotesize
\begin{split}
&\frac{R}{1-R}\left[\left(\sum_{T=1}^{\infty} q^{T}\binom{T-2}{z-2}\right)\right]\\ &- \frac{1}{1-R}\left[\left(\sum_{T=1}^{\infty} \left(Rq\right)^{T}\binom{T-2}{z-2}\right)\right]\\
= & \frac{R}{1-R}\left(\left(-1\right)^{z}q+\frac{q^{z}}{{\left(1-q\right)^{z-1}}}\right)\\ &- \frac{1}{1-R}\left(\left(-1\right)^{z}Rq+\frac{\left(Rq\right)^{z}}{{\left(1-Rq\right)^{z-1}}}\right)\\
= & \frac{1}{1-R}\left(\left(-1\right)^{z}Rq+R\frac{q^{z}}{{\left(1-q\right)^{z-1}}}\right)\\ &- \frac{1}{1-R}\left(\left(-1\right)^{z}Rq+\frac{\left(Rq\right)^{z}}{{\left(1-Rq\right)^{z-1}}}\right)\\
= & \frac{1}{1-R}\left(R\frac{q^{z}}{{\left(1-q\right)^{z-1}}} - \frac{\left(Rq\right)^{z}}{{\left(1-Rq\right)^{z-1}}}\right)\ .
\end{split}
\end{equation}


Putting this back into equation \eqref{eq:split_eq2}, combined with equation \eqref{eq:firstterm}, together with the factor of $\frac{1}{z}\left(\frac{p}{1-p}\right)^z$ gives that the average noise parameter is given by
\begin{equation}\footnotesize
\begin{split}
\langle\lambda\rangle =&\frac{1}{z}\left(1+\left(\frac{1-q}{q}\right)^z\frac{1}{1-R}\left(R\frac{q^{z}}{{\left(1-q\right)^{z-1}}} - \frac{\left(Rq\right)^{z}}{{\left(1-Rq\right)^{z-1}}}\right)\right)\\
=&\frac{1-Rq}{\left(1-R\right)z}\left(1-\left(\frac{R-Rq}{1-Rq}\right)^z\right)\
\end{split}
\end{equation}
This completes the proof.

\vspace{-0.4cm}
\section{Validation of the Approximation in \eqref{eq_approximation}} \label{appendix_validation_approximation}\vspace{-0.1cm}
We conduct numerical experiments to validate the approximation in \eqref{eq_approximation}. Particularly, we compare the exact numerical average distillation fidelity, $\langle F_{\text{dis}}(F_{\text{input}},z_i)\rangle$, and probability of success, $\langle P_{\text{dis}}(F_{\text{input}},z_i)\rangle$ with their corresponding approximate averages proposed in \eqref{eq_approximation}, which are $F_{\text{dis}}(\langle F_{\text{input}}\rangle,z_i)$ and $P_{\text{dis}}(\langle F_{\text{input}}\rangle,z_i)$, respectively. Since the input fidelity to both distillation fidelity and probability of success expressions depends on the control variables  $\theta_i, x_i, y_i,$ and $z_i, \forall i \in\mathcal{U}$, we compare the corresponding averages as the probability of successfully generating an e2e entangled link between the QS and user $i\in\mathcal{U}$, that is $P_{\text{in}}(\theta_i,x_i)$, for all possible $z_i$ values. 

Due to limited space, we only show the resulting averages for the case when $z_i=5$, in Fig. \ref{fig_validation_approximation}. We numerically validate, by averaging over 1,000,000 SIS instances, that the approximation is proven to be accurate with a maximum percentage error of below $5\%$. Particularly, the approximation is very accurate with an error below $1\%$, except for the range $P_{\text{in}}(\theta_i,x_i)\in(0.1,0.2)$, where it marginally increases. To be more specific, the largest achieved percentage error was $4.45\%$ in the e2e fidelity, which is below $5\%$, and it occurred at $P_{\text{in}}$ values of $0.135$ in the entanglement distillation scheme with $z_i = 5$, as shown in Fig. \ref{fig_validation_approximation}. 

\begin{figure}
\begin{center}
\includegraphics[width=\columnwidth]{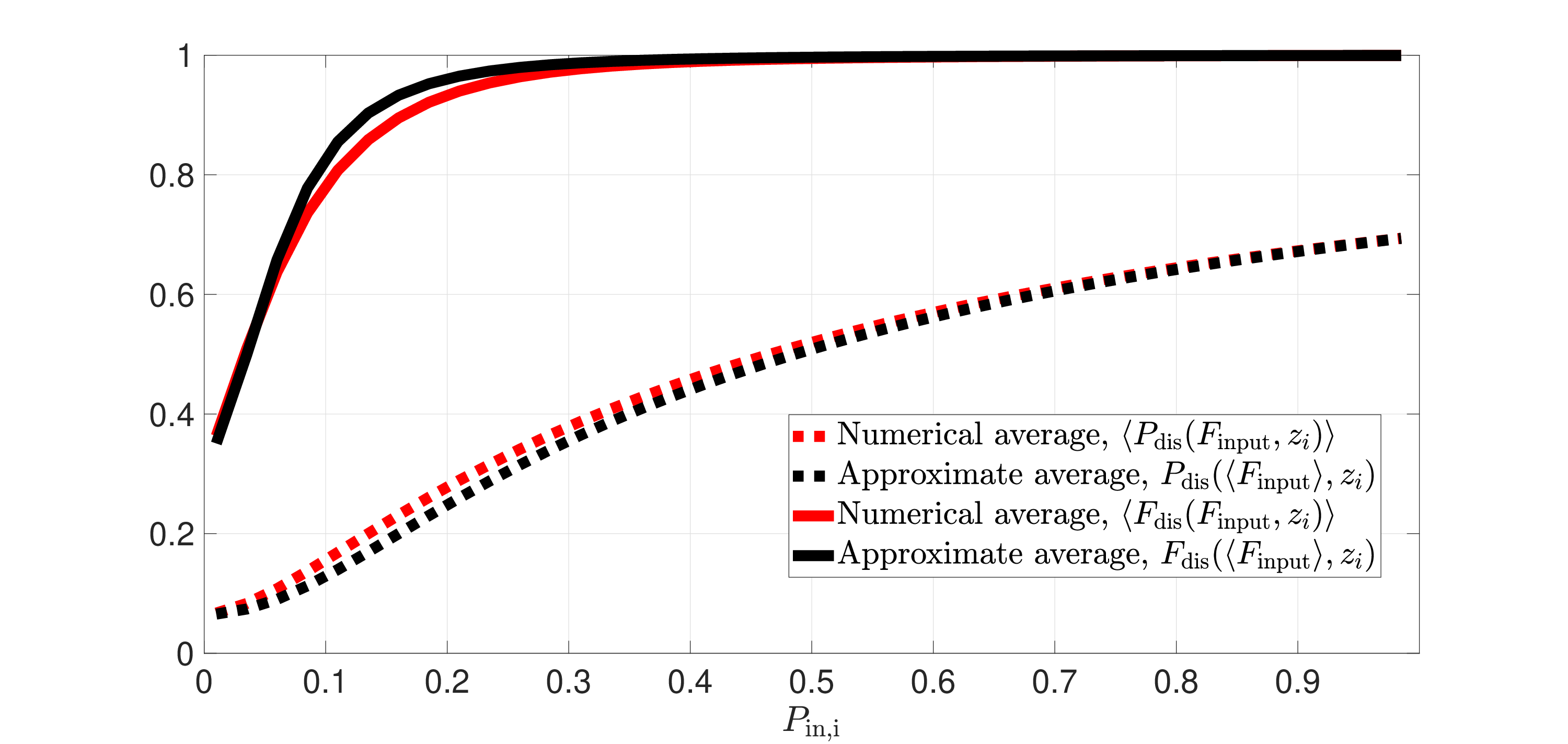}\vspace{-0.3cm}
\caption{A comparison between the exact numerical average and considered approximate average entanglement distillation fidelity and probability of success. The comparison is performed as the probability of successfully establishing an e2e entangled link between the QS and user $i\in\mathcal{U}$ is varied for the case of $z_i=5$.}\vspace{-0.15cm}
\label{fig_validation_approximation}
\end{center}
\end{figure}

\ifCLASSOPTIONcaptionsoff
  \newpage
\fi

\vspace{-3mm}\begin{spacing}{0.99}
\bibliographystyle{ieeetr}
\def\baselinestretch{0.86}
\bibliography{references}\vspace{-2mm}
\end{spacing}

\end{document}